\def\year{2019}\relax

\documentclass[letterpaper]{article} %DO NOT CHANGE THIS
\usepackage{aaai19_mod}  %Required
\usepackage{times}  %Required
\usepackage{helvet}  %Required
\usepackage{courier}  %Required
\usepackage{url}  %Required
\usepackage{graphicx}  %Required
\usepackage{amsmath,amssymb,amsthm,amsfonts}
\usepackage[linesnumbered,algoruled,boxed,lined,algoruled,boxed,lined]{algorithm2e}

\newtheorem{proposition}{Proposition}
\newtheorem{theorem}{Theorem}
\newtheorem{definition}{Definition}
\newtheorem{lemma}{Lemma}

\usepackage{graphicx}
\usepackage{epstopdf}

\usepackage{url}
\usepackage{enumitem}
\usepackage{wrapfig}
\nocopyright

\usepackage{xcolor}
\ifodd 0
\newcommand{\rev}[1]{{\color{blue}#1}} %revise of the text
\newcommand{\com}[1]{\textbf{\color{red}(COMMENT: #1)}} %comment of the text

\newcommand{\adv}[1]{{\color{red}#1}}

\else
\newcommand{\rev}[1]{#1}
\newcommand{\com}[1]{}

\newcommand{\adv}[1]{#1}
\fi
\title{Adversarial Contract Design for Private Data Commercialization}

\author{
Parinaz Naghizadeh$^1{^*}$ and 
Arunesh Sinha$^2$\thanks{Parinaz and Arunesh equally contributed to this work.}
\\ 
$^1$ Electrical and Computer Engineering, Purdue University \\
$^2$ Computer Science and Engineering, University of Michigan\\
parinaz@purdue.edu,
arunesh@umich.edu
}

\begin{document}

\maketitle

\begin{abstract}
The proliferation of data collection and machine learning techniques has created an  opportunity for commercialization of private data by data aggregators. In this paper, we study this data monetization problem using a contract-theoretic approach. Our proposed adversarial contract design framework accounts for the heterogeneity in honest buyers' demands for data, as well as the presence of adversarial buyers who may purchase data to compromise its privacy. We propose the notion of \emph{Price of Adversary} $(PoAdv)$ to quantify the effects of adversarial users on the data seller's revenue, and provide bounds on the $PoAdv$ for various classes of adversary utility. We also provide a fast approximate technique to compute contracts in the presence of adversaries. 
\end{abstract}

\section{Introduction}\label{sec:intro}

The large-scale adoption of data-driven decision making by businesses has led to a boom in big data collection and analysis techniques. With increasing  amount and demand for data, companies have found a business opportunity in offering data-based services to other companies, or selling their data to interested parties \cite{thomas16,spiekermann15}. Interest in data monetization is evidenced by the rise of \emph{data marketplaces}, where firms and individuals can buy, sell, or trade, second or third party data. Examples include Salesforce's Data Studio, Oracle's BlueKai, and Adobe's Audience Marketplace.
Data commercialization faces many challenges, including IP protection, liability, pricing, and preserving  privacy~\cite{thomas16}. In this paper, we focus on the latter two challenges of pricing and privacy. 

The challenge of pricing refers to the fact that to accommodate diverse demands, data sellers offer different plans and pricing to their customers. 
%For instance, data may be categorized based on attributes such as source industry, geographic location, and demographics. Data may further be \emph{priced-per-record}, to allow different customers to request different amounts of data based on their needs. 
Even with identical data, customers may derive different benefits from utilizing it, e.g., due to different expertise, or how this data complements the customer's existing knowledge. Therefore, to maximize revenue, the data seller should account for this demand diversity by packaging its data accordingly. 
%Multi-tier architecture used for database access, e.g. there is this three-tier database access in  http://repository.ubn.ru.nl/bitstream/handle/2066/35187/35187.pdf?sequence=1; where the middle-tier controls access management to allow access to the database from multiple users. middle tier can be viewed as the operator, giving prioritized access levels.
Further, despite its revenue benefits, data commercialization has to overcome the challenge of limiting privacy risks for the data subjects in the database. Specifically, adversarial buyers can request access to the database, attempting to compromise the privacy of the data subjects.
Therefore, data sellers should  account for this risk when designing and pricing data plans.

In this paper, we take a contract theoretical~\cite{mas95} approach to address both the aforementioned pricing and privacy challenges of data commercialization by proposing the design of a set of contracts with varying privacy levels. Contract theory, in the classic context of pricing of goods, is the study of principal-agent problems, in which the principal (here, the data seller) designs a set of contracts with varying consumption level so as to extract maximum revenue from agents (buyers) with unknown types. 
%In general, contracts are designed in the presence of two types of informational asymmetry between the principal and the agents: moral hazard (unobservable actions by agents) and adverse selection (unobservable types of agents). In this paper, we focus on the problem of adverse selection, by assuming that the type of users (legitimate or adversarial, and the particular benefit of a given user from the data) are unknown to the data owner. 
We study the problem of pricing a bundle of database queries at different privacy levels with the aim of (a) maximizing revenue by offering different prices for varying privacy levels in order to accommodate the diversity of demands for the query bundle, and (b) accounting for the risks from adversarial users by modifying the contracts' pricing accordingly. We use the well accepted \emph{$\epsilon$-differential privacy} concept as the measure of privacy \cite{dwork2008differential}. %, where the value of $\epsilon$ captures the level of privacy loss.
We make an effort to keep our design practical by attempting to adhere to practices already in place in data marketplaces (see Sections \ref{sec:background} and \ref{sec:model}).

% In particular,
\textbf{Technical contributions:} (1) Existing contract theory results suggest that given $n+1$ types of agents ($n$ types of honest buyers based on their diversity of demand, and an adversarial type), the principal should design up to $n+1$ contracts. We show that the data owner will offer at most $n$ contracts. In other words, it is optimal for the data owner to avoid the impractical option of designing a contract for the adversary; (2) we incorporate post-hoc fines (in case of privacy breach) in the pricing of query bundles, and analyze their effect on the contract design problem, showing that fines can be helpful in reducing loss due to the adversarial users in many situations; (3) we propose the notion of \emph{Price of Adversary} ($PoAdv$) to quantify the loss incurred by the data owner due to the presence of adversarial data buyers. We show that while $PoAdv$ can be unbounded in the worst case, it is possible to bound the $PoAdv$ for a large class of problems; and (4) we provide a fast approximate technique to compute the contracts in presence of adversaries.
%characterize properties of the optimal contracts, and show that the optimal contracts in the presence of adversaries retain several of the features of the optimal contracts attained from the non-adversarial contract design problem; , and (4) we extend our contract mechanism to allow for the use of \emph{security deposits} in the contracts, and analyze their effect on the contract design problem as well as on the $PoAdv$, showing that security deposits can be helpful in reducing loss due to adversaries in many situations. 
All omitted and full version of proofs are in the appendix. %an online appendix: \url{https://goo.gl/EYxnwY} 

% Our main contributions are therefore the following: 

% $\bullet$ We propose a contract-theoretical analysis of data commercialization in the presence of adversarial types. We show that it is never optimal for the operator to screen for adversarial types, and characterize properties of the optimal contracts in the presence of adversaries.    

% $\bullet$ We propose the notion of price of adversary to quantify the effect of adversarial types on a data owner's revenue potential. We find bounds on the $PoAdv$, and propose a heuristic modification of the collection of contracts in the presence of adversaries. 

%{\bf Paper organization:} In Section~\ref{sec:model}, we present the model and preliminaries on the effect of adversarial types in the contract design problem. We quantify the effect of adversaries through the study of the price of adversary in Section~\ref{sec:advcontract}. We extend the model to include contracts with security deposits in Section~\ref{sec:SDM}, and  conclude in Section~\ref{sec:conclusion}. 

\section{Background}\label{sec:background}
\textbf{Database marketing examples}: 
Currently, the two industries leading database marketing are data brokers (who mine and sell consumer data to businesses), and data marketplaces (which provide a platform for buying, selling, and trading data). We elaborate upon typical privacy guarantees offered by each with an example. Among data brokers, Acxiom, one of the largest brokers worldwide, states that they maintain ``privacy compliant data" through data encryption and secure data management techniques~\cite{acxiom}. The user service agreement of Salesforce Data Studio~\cite{salesforce} on the other hand, provides more detailed information about their market structure. For instance, Salesforce states that 
they use ``unique user identifiers (user IDs) to help ensure that activities can be attributed to the responsible individual'', and that security logs are kept ``in order to enable security reviews and analysis."  
Our model in Section \ref{sec:model} takes the availability of these monitoring techniques into account. It is clear that following such safe practices is imperative when dealing with private information, e.g., as evidenced by the recent Cambridge Analytica case~\cite{cambridge}. 

%Consider a large social network, say Facebook. Facebook collects data from its users who are compensated in form of the services provided by Facebook. Facebook provides buyers with query access to this data. However, Facebook runs a risk of economic loss due to misuse of the query results by the buyers, as was the case in the recent Cambridge Analytica case~\cite{}. Building on this example, in the next section we provide a design for this data market that follows known safe practices in privacy. \com{to be edited} 
%\footnote{\com{I'm going to be changing this a bit, both to add the example of data and query that we talked about, and also as apparently technically FB does not sell the data/queries; companies say what demographic they want, then FB does the query (without giving them the result) and just gives them targeted  advertising on the platform (according to e.g. this \url{https://www.recode.net/2018/4/11/17177842/facebook-advertising-ads-explained-mark-zuckerberg} and \url{https://www.nytimes.com/2018/03/19/technology/facebook-cambridge-analytica-explained.html}). The CA case was users granting access to another app to collect their data through FB, and then the data being sold in breach of FB rules. Anyway, just wanted to say that working on changing this accordingly. Let me know if you also think of suggestions of other companies that'd be good to include.}}

\textbf{Differential privacy}: A popular formalism of privacy loss due to adversarial queries from statistical databases is that of \emph{differential privacy (DP)} \cite{dwork2006calibrating,dwork2008differential}. Formally, let $\mathcal{K}$ be a randomized algorithm used by a data owner to release answers to queries from a database, and consider two databases $\mathcal{D}_1$ and $\mathcal{D}_2$ that differ in exactly one entry (row). Then, $\mathcal{K}$ is $\epsilon$-differentially private ($\epsilon$-DP) for $\epsilon \geq 0$ if for any output $\mathcal{O}$ of the algorithm,
\begin{align}
Pr(\mathcal{K}(\mathcal{D}_1)\in\mathcal{O}) \leq \exp(\epsilon) \cdot Pr(\mathcal{K}(\mathcal{D}_2)\in\mathcal{O})~.
\label{eq:df-def}    
\end{align}
In words, $\epsilon$-DP requires that the output of $\mathcal{K}$ remains sufficiently unaffected (as quantified by $\epsilon$), whether or not a single data subject's data is included in the database. The choice of $\epsilon$ determines the privacy loss due to $\mathcal{K}$, with lower $\epsilon$ corresponding to better privacy. Note that this privacy guarantee is independent of any  auxiliary information available to an adversary \cite{dwork2008differential}. 

For continuous-valued queries, a method for achieving differential privacy is the introduction of carefully selected random noise in the responses. Specifically, let $f$ be a query function, returning the true value $f(D)$ on database $D$. In order to guarantee $\epsilon$-DP, an algorithm $K$ can introduce additive Laplacian noise, returning instead $f(D)+\text{Lap}(\Delta f/\epsilon)$, where $\Delta f$ is the sensitivity of the query function \cite{dwork2006calibrating}. While this approach limits the privacy loss to within $\epsilon$, it also decreases the utility of the queries to  honest buyers by adding noise that increases in $\frac{1}{\epsilon}$. In Section \ref{sec:model}, we formalize buyers' sensitivity to their queries' accuracy, and consequently, their willingness to pay for more accurate answers. The seller's contract design should balance this tradeoff with the increased privacy loss from adding less noise.

\section{Model}\label{sec:model}
We study the problem of designing a set of contracts for \emph{buyers} requesting access to a database managed by a \emph{seller}. \rev{We assume that the seller has already acquired data from subjects and compensated them using a one-time monetary payment or a free service (like a phone app).} We use he/his for buyers and she/her for the seller. 

\textbf{Queries:} There are multiple (and finite) types of potential statistical queries that can be made from the database, denoted by the set $\mathcal{Q}$. 
The seller offers bundles consisting of a subset of these query types for purchase, with the restriction that any buyer can choose at most one bundle. A bundle is identified by the set $\{Q_1, \ldots, Q_k~|~Q_i \in \mathcal{Q}\}$. The seller designs these bundles based on historical or external information about the types of different buyers, so that every buyers' requirement is met by one of the bundles. Further, for any bundle, the seller limits the number of queries of each type $Q_i$ in the bundle to one (i.e., each bundle is a subset of distinct query types). %This is not restrictive as the buyer can always cache the queries' answers at her end.\footnote{\com{if a buyer can cache query answers, then wouldn't we have to consider the privacy loss of that?}} 
This follows recommended practices in differential privacy, since allowing multiple queries inevitably degrade privacy guarantees (see also Section \ref{sec:related}). We also posit that the seller verifies the identity of  buyers, in order to keep track of the buyer's query purchases, and to investigate a privacy attack if it occurs. Further, we posit that the seller, via her service agreement, restricts buyers from faking identifies by imposing substantial post-hoc fines. %, and imposes a large post-hoc fine for faking identities.

%provides a list of queries that any buyer can issue. We impose a natural assumption that a buyer issues a single query

\textbf{Contracts}: For each bundle $\{Q_1, \ldots, Q_k\}$, the set of possible contracts are determined by the parameters $(p, \epsilon, s)$, with $p\in \mathbb{R}_{\geq 0}$ denoting the price to be paid by the buyer. The privacy levels are assumed to be bounded and normalized such that {$\epsilon\in[0,1]$}, with $\epsilon$ specifying the bound $\epsilon \geq \epsilon_1 + \ldots + \epsilon_k$, where $\epsilon_i$ is used to determine the (Laplace) noise added to the answer of the  query of type $Q_i$; the buyer is free to request any $\epsilon_1, \ldots, \epsilon_k$ within the $\epsilon$ bound, \rev{with higher $\epsilon$ corresponding to less noisy responses.} % when she chooses $(p, \epsilon)$.
Lastly, $s$ denotes the post-hoc fine to be paid if the buyer is found misusing the query answer.
%Each buyer is interested in issuing one query\footnote{The assumption of one query per buyer is an assumption based.}. 
%Each contract is determined by a pair $(p, l)$, with $p\in \mathbb{R}_{\geq 0}$ and $l\in[0,1]$ denoting its price and access level, respectively. 

\textbf{Buyers}: We assume that buyers belong to one of two possible classes: \emph{honest} or \emph{adversarial}.

\emph{Honest buyers}: Honest buyers do not misuse query answers, and hence generate revenue for the operator when purchasing contracts. Each honest buyer for a given bundle has a type $i \in \Theta:=\{1, \ldots, n\}$, determining his benefit from the database. 
In particular, an honest buyer of type $i$ purchasing contract $(p, \epsilon,s)$ derives a \emph{benefit} $b_i(\epsilon):[0,1]\rightarrow \mathbb{R}_{\geq 0}$ from accessing the system. \rev{This function includes direct gain from the data, as well as the cost of hedging against the risk of potential direct attack on the buyer.} %(e.g., in form of premium for  cyber-insurance~\cite{bohme2010modeling}).}
We impose natural conditions on the benefit functions (as is standard for demand functions) $b_i(\cdot)$: that the overall benefit increases with larger $\epsilon$ (monotone non-decreasing) and satisfies  diminishing returns (concavity), with $b_{i}(0) = 0$.  Most large organizations estimate demand functions and types of buyers from past buyers' activity, and insurance premiums are known; hence, we assume these functions are known. \rev{Further, $b_{i}(\epsilon)\leq b_{{i+1}}(\epsilon), \forall \epsilon, \forall i$; that is, higher types derive further benefit from the same noise level, e.g., due to their expertise or the relevance of the data to their tasks.} 

An honest buyer also has a $\gamma$ probability of suffering an attack himself and causing inadvertent misuse of the query answer, which results in an expected $\gamma s$ loss for him as per the contract terms. Thus, an honest buyer's overall expected utility in its interaction with the seller is given by $u_{i}(p, \epsilon, s)=b_{i}(\epsilon) - p - \gamma s$. %\com{ordering of honest buyers?}

\emph{Adversarial buyers}: An adversarial buyer seeks to access the database with the goal of compromising its privacy. Formally, an adversarial buyer purchasing a bundle through a contract $(p, \epsilon, s)$ derives a benefit $C(\epsilon): [0,1]\rightarrow \mathbb{R}_{\geq 0}$ from an attack on the system, with overall adversary utility given by $u_A(p, \epsilon, s)=C(\epsilon)-p - s$. This attack results in a cost $C(\epsilon)$ for the seller. Further, we assume $C(\cdot)$ is monotone increasing and convex, with $C(0) = 0$; intuitively, higher $\epsilon$ (lower noise) lead to costlier attacks for the seller, with the severity increasing as the noise decreases. Such convexity has also been noted in literature, e.g., a \rev{recent work~\cite{hsu2014differential}} proposes the cost for seller to be proportional to $\exp(\epsilon) - 1$. Figure~\ref{fig:b-C-P} shows an example of $C$ and $b_i$. 

We assume that a privacy attack is ultimately discovered, and the seller can track the buyer responsible for the attack. \rev{The seller may have to compensate data subjects after a privacy attack (due to lawsuits), which can be partially recovered from the post-hoc fine for data misuse.} Note that we have assumed that the adversary cannot cause privacy loss beyond the given $\epsilon$ of the bundle by combining the outputs of multiple queries of the same type, as the seller restricts the number of queries per type to one. Further, large post-hoc fines for faking identities prevent the rational adversary from faking identities and attempting to purchase two or more bundles. However, the post-hoc fine for data misuse cannot be set too large as this fine affects the honest buyers, and hence the seller's revenue, due to potential attacks on honest buyers. Therefore, our goal is to study the optimal choice of fines for data misuse so as to deter adversarial buyers while maintaining the demand from honest buyers.  

\begin{wrapfigure}{l}{0.25\textwidth}
% \hspace*{-0.02in}%
\centering
\includegraphics[width=0.3\textwidth]{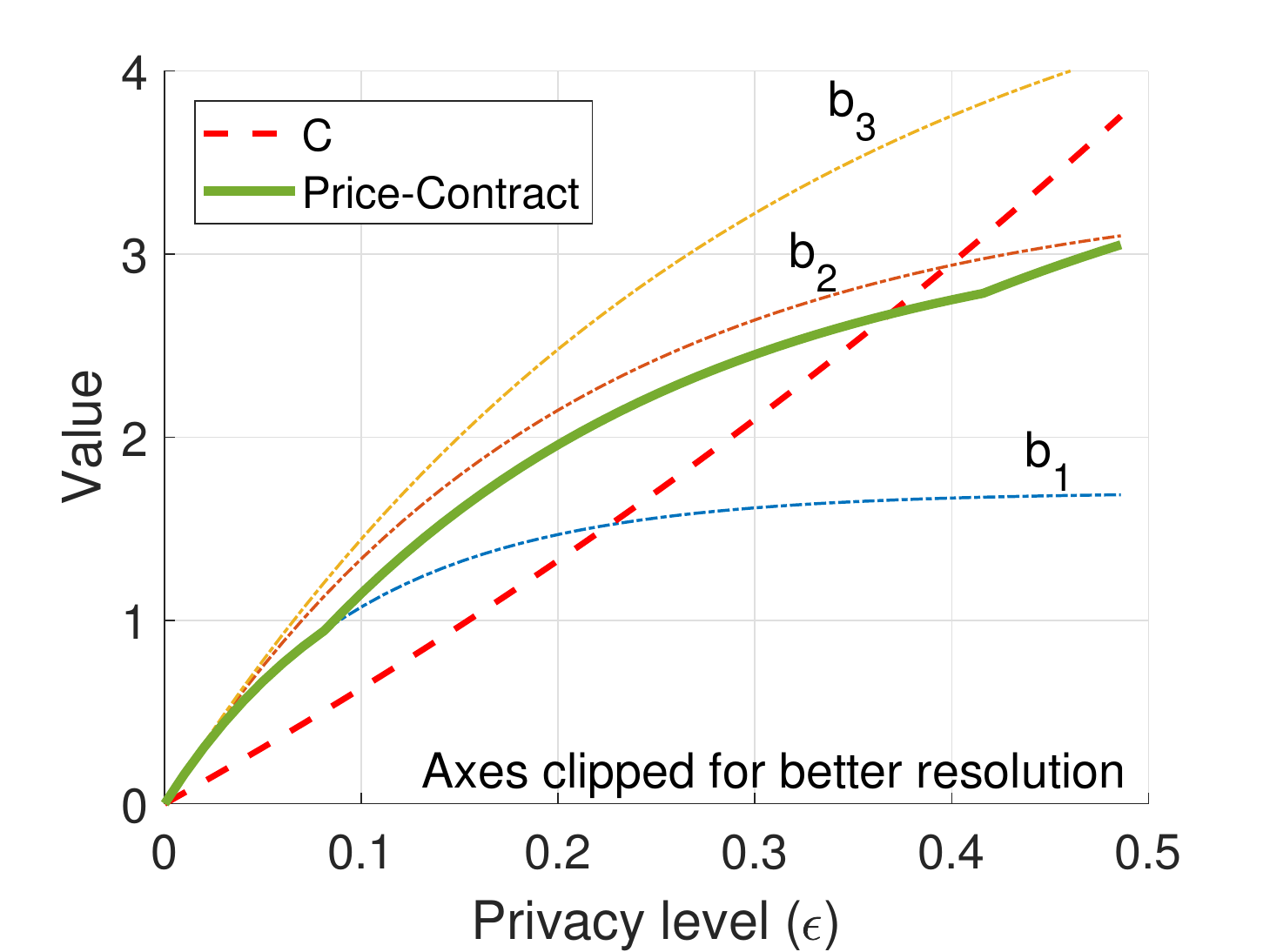}
\caption{Three Benefit fns $b_i(\epsilon)=i(1 - \exp(-\frac{10}{i}\epsilon))$ for $i=1,2,3$, adversary cost $C(\epsilon)=6(\exp(\epsilon)-1)$, and non-adversarial price-contract curve.}
\label{fig:b-C-P}
\end{wrapfigure}

\subsection{Seller's revenue optimization problem}
We now analyze the seller's contract design problem, with one contract $(p,\epsilon, s)$ for each offered bundle. As all rational buyers choose only one bundle due to the marketplace design, these contracts are independent. Therefore, for the rest of this paper, we restrict attention to a given bundle. %; the same analysis applies to any other bundle. 

Let $\rho$ denote the fraction of adversarial buyers, which is estimated by the seller (conservatively, the seller can estimate $\rho$ to be at most a maximum value). For the honest buyers, let $\{q_i\}_{i\in\Theta}$ denote the fraction of the honest buyers of type $i$. These fractions can be estimated from historical data. %\footnote{\rev{In the absence of historical data, the seller can attempt to estimate $\{q_i\}$, as well as $\rho$, in an online fashion through repeated contract design. This is beyond the scope of the current paper, and remains an interesting direction of future work.}}
 The seller aims to maximize her revenue. Nevertheless, she can not observe individual buyers' types when selling a contract. Consequently, she has to design contracts while balancing two goals: deriving the maximum possible profit from honest types, while limiting the adversarial type's cost to the system. 

%The operator can take steps during the design of the contracts to distinguish, or \emph{screen}, the agents, so as to achieve the aforementioned balance. 

In classic contract theory, following the \emph{revelation principle}~\cite[Proposition 14.C.2]{mas95}, it is known that it is enough to offer at most $n+1$ contracts  when the number of buyer types is $n+1$. \footnote{Depending on the distribution of buyers' types, it may be optimal to offer the same contract to adjacent types (\emph{pooling} contracts), instead of separate contracts for each type (\emph{separating} contracts).} 
Each agent then selects his intended contract if it satisfies the agent's \emph{individual rationality} (IR) and \emph{incentive compatibility} (IC) constraints. The IR constraint requires that the agent attains higher utility from purchasing the contract compared to opting out. The IC constraint imposes the condition that an agent of type $i$ prefers his intended contract over that of any type $j\neq i$. 

Formally, for our contract design problem, consider the $n+1$ user types consisting of the $n$ honest buyers and the adversarial type. Let the contract of type $i\in\Theta$ be $(p_i, \epsilon_i, s_i)$ and that for the adversary be $(p_A,\epsilon_A, s_A)$. Assume, wlog, that the utility of opting out of purchasing contracts is zero. Then, the IR constraint of an honest buyer of type $i$, denoted by $\textit{IR}_{i}$, is given by
$u_{i}(p_i, \epsilon_i, s_i) \geq 0$. Similarly, $IR_A$ is $u_{A}(p_A, \epsilon_A, s_A) \geq 0$.
%\label{eq:IR-theta-i}
Type $i$'s IC constraints are given by $
u_{i}(p_i, \epsilon_i, s_i) \geq u_{i}(p_j, \epsilon_j, s_j), \forall j\neq i~$
where the $j^{th}$ constraint is denoted by $\textit{IC}_{i,j}$ and $
u_{i}(p_i, \epsilon_i, s_i) \geq u_{i}(p_A, \epsilon_A, s_A)$ which is denoted as $\textit{IC}_{i,A}$. Similarly, the $IC_{A,i}$ constraints can be defined for the adversary. 

The seller's goal is to maximize her revenue $R\big((p_i, \epsilon_i, s_i)_{i \in \Theta}, p_A,\epsilon_A, s_A \big) = (1 - \rho) \big(\sum_{i=1}^n q_i (p_i + \gamma s_i) \big ) + \rho (p_A + s_A - C(l_A))$. \rev{However, the seller only has steady revenue over time from $p_i$; $\gamma s_i$ provides randomly varying revenue over time. Thus, we impose the practical constraint that $p_i \geq (1- \phi) (p_i + \gamma s_i)$, which says that a large fraction $1- \phi$ of revenue arrive steadily over time. We name this the \emph{steady revenue} $SR_i$ constraint.}
Therefore, the seller's contract design problem can be formally stated as the following optimization: 
\begin{eqnarray*}
\max_{(p_i, \epsilon_i, s_i)_{i \in \Theta}, p_A,\epsilon_A, s_A} & R\big((p_i, \epsilon_i, s_i)_{i \in \Theta}, p_A,\epsilon_A, s_A \big)~\\
\mbox{subject to} & IR_i, SR_i \; \forall i \; \mbox{ and } \; IC_{i,j} \; \forall i,j \mbox{ and}\\
&  IR_A \; \mbox{ and } IC_{i,A}, IC_{A,i} \; \forall i \mbox{ and }\\
& p_i,\epsilon_i, s_i \geq 0 \; \forall i \mbox{ and }  p_A,\epsilon_A, s_A \geq 0
\end{eqnarray*}

%However, later we show that for our problem, it is optimal to have at most $n$ contracts, one for each legitimate user. As a result, adversarial users will be selecting one of the legitimate contracts (if any). We begin by providing a background on classic contract theory, which also introduces standard contract theory concepts and terminology which we will use throughout.

%In general, the operator can propose a collection of contracts, with each contract optimally designed so as to be preferred by (at least) one user type. In the following, we will show that while the operator may benefit from offering separate contracts for legitimate types, it is not possible/profitable to design a separate contract specifically for adversarial users.  Hence, the operator should optimally modify the legitimate types' contracts to limit the risks from adversarial users.  

\subsection{No need for an adversary-specific contract}
The contract design problem above includes a contract $(p_A,\epsilon_A, s_A)$ for the adversary. While the formulation is mathematically sound and consistent with the revelation principle, this seems an odd design choice as the adversary reveals his type just by choosing this contract. We show that, as intuitively expected, it is in fact not required for the seller to design an adversary-specific contract. % and still achieve the same revenue.

\begin{lemma}\label{lemma:Ncontracts}
{The seller should offer at most $n$ contracts/bundles. In particular, it is never optimal to offer an adversary-specific contract/bundle.}
\end{lemma} 
\begin{proof}
We show this by contradiction. Assume the seller treats the adversarial buyer as the $(n+1)$-th type, and offers a contract $(p_A, \epsilon_A, s_A)$ satisfying all (honest and adversarial) buyers' IR and IC constraints. By $\textit{IR}_A$, this contract satisfies $C(\epsilon_A)-p_A-s_A\geq 0$; that is, it will impose a loss $p_A+s_A-C(l\epsilon_A)\leq 0$ on the seller's revenue. Further, by the $\textit{IC}_{A,i}$ constraints, $C(\epsilon_A)-p_A-s_A\geq C(\epsilon_i) - p_i - s_i$; that is, had the adversary purchased any of the legitimate buyers' contracts, he would have imposed a smaller cost on the seller's revenue. As the seller is a profit-maximizer, we conclude that such contract $(p_A, \epsilon_A, s_A)$ should not be part of an optimal collection of contracts. 
\end{proof}

% We conclude that to limit the adversaries' cost, the operator has to optimally modify the collection of contracts that are intended to screen the legitimate users' types. 
%
%\subsection{The adversarial contract design problem}\label{sec:advcont}
Given the above lemma, the contract design problem in the adversarial setting is to design contracts $(p_i, \epsilon_i, s_i)_{i \in \Theta}$ in order to maximize the revenue of the operator:
\[
(1 - \rho) \big(\sum_{i=1}^n q_i (p_i + \gamma s_i) \big ) + \rho (p_Z + s_Z - C(\epsilon_Z))~,
\]
where $Z \in \{0,1,\ldots,n\}$ is the contract chosen by the adversary,
subject to IR and IC constraints for all honest buyers in choosing their contract $i$ and the adversary in choosing $Z$. For the special case of the adversary not choosing any contract, we designate $Z=0$ with $p_0 = s_0 = \epsilon_0 = 0$. Observe that $Z$ is a variable, and thus, the revenue maximizing problem is a bi-level optimization problem. However, following the standard technique of introducing an additional variable to formulate a zero-sum problem as a linear program, we formulate the revenue maximization problem in the adversarial setting using variable $r_A$ as follows:
\begin{eqnarray*}
\max_{(p_i, \epsilon_i, s_i)_{i \in \Theta}, r_A} & (1 - \rho) \big(\sum_{i=1}^n q_i (p_i + \gamma s_i) \big ) + \rho (- r_A)~\\
\mbox{subject to} & IR_i, SR_i \; \forall i \; \mbox{ and } \; IC_{i,j} \; \forall i,j \mbox{ and}\\
&  r_A \geq C(\epsilon_i)  - p_i - s_i \; \forall i \; \mbox{ and }\\
& p_i,\epsilon_i, s_i \geq 0 \; \forall i \; \mbox{ and } \;  r_A \geq 0
\end{eqnarray*}
For our described marketplace, one can further consider the corresponding non-adversarial setting, in which the seller solves the contract design problem in the absence of any adversarial considerations. This non-adversarial contract design problem is given by:
\begin{eqnarray*}
\max_{(p_i, \epsilon_i, s_i)_{i \in \Theta}} &  \sum_{i=1}^n q_i (p_i + \gamma s_i)  \\
\mbox{subject to} & IR_i, SR_i \; \forall i,~ IC_{i,j} \; \forall i,j, \mbox{ and, }  p_i,\epsilon_i, s_i \geq 0 \; \forall i 
\end{eqnarray*}

%The adversary will select contract $Z$ if and only if $C(l_Z) - p_Z -s_Z \geq 0$ and $C(l_Z) - p_Z - s_Z \geq C(\epsilon_i) - p_i - \gamma s_i \; \forall i \neq Z$. 
We next study these two contract design problems to characterize the effects of the presence of adversarial types on the optimal contracts' properties and the seller's revenue.

\section{Analysis of Adversarial Contracting}\label{sec:advcontract}

In classic contract theory, when solving for the optimal contracts, the functions $b_i$ are often assumed to satisfy a condition known as the \emph{single crossing property (SCP)}, which in turn implies the strict increasing differences (ID) property. Throughout our analysis, we will only require the (weaker) condition of (non-strict) ID property on the benefit functions $b_i$, as defined below:
\begin{definition}[Increasing Differences]
The functions $b_{i}$ satisfy the (strict) increasing differences property if for any $\epsilon'> \epsilon$, $b_i(\epsilon') - b_i(\epsilon)$ is (strictly) increasing in the type $i$. 
\end{definition}

The above condition %has been found to natural for demand functions, and
is a natural assumption on demand functions, and has been used extensively in the contract theory literature starting from the seminal work by~\cite{maskin1984monopoly}. The $b_i$ functions shown in Figure~\ref{fig:b-C-P} satisfy ID. This condition also allows for significant simplification of the classical contract theory optimization problem.
Our first, somewhat surprising  result is that, even in the adversarial contract regime with post-hoc fines, the contracts will satisfy a set of constraints akin to those of non-adversarial settings. % \clar{with the even weaker assumption of ID instead of strict ID:}
\begin{theorem} \label{advcontractproperties}
Assuming that the functions $b_i$ satisfy ID, the optimal contracts (in the presence of adversarial types) $(p^*_1,\epsilon^*_1, s^*_1), \ldots, (p^*_n,\epsilon^*_n, s^*_n)$ satisfy the following:
\begin{enumerate}
\item Monotonicity: $\epsilon^*_{i+1} \geq \epsilon^*_i, \forall i$. %\com{again, reversed, right? Or are we indexing the other way?}
\item Constraint set reduction: $\textit{IR}_i$ for $i > 1$ and $\textit{IC}_{i,j}$ for $j \neq i-1$ are redundant at the optimal contracts.
\item $\textit{IR}_{1}$ is tight: as a result, $p^*_1 + \gamma s^*_1 = b_1(\epsilon^*_1)$.
\item $\textit{IC}_{i+1,i}$ is tight for all $i$: as a result for $i > 1$, 
$$p^*_{i} + \gamma s^*_{i} = b_{i}(\epsilon^*_{i}) - \sum_{j=1}^{i-1} \big( b_{j+1}(\epsilon^*_{j}) - b_{j}(\epsilon^*_j) \big )~.$$
\end{enumerate}
% following, 
% \begin{enumerate}
% \item Monotonicity: $l^*_{i+1} \geq l^*_i, \forall i$. 
% \item Constraint set reduction: $IR_i$ for $i > 1$ and $IC_{i,j}$ for $j \neq i-1$ are redundant at the optimal contracts.
% \item No rent for the low type: $IR_{1}$ is tight, i.e., $p^*_1 = b_1(l^*_1)$.
% \item Information rent for other types: $IC_{i+1,i}$ is tight for all $i$, and hence for $i > 1$, 
% $$p^*_{i+1} = b_{i+1}(l^*_{i+1}) - \sum_{j=1}^{i} \big( b_{j+1}(l^*_{j}) - b_{j}(l^*_j) \big )~.$$
% \end{enumerate}
\end{theorem}
\begin{proof}[Proof Sketch]
We first establish the monotonicity of noise levels at the optimal contracts using the (non-strict) ID condition of the benefit functions. 
Next, we show how to considerably refine the constraint set (point 2) and derive the price-benefit relations (points 3-4). These arguments are based on contradiction: had any of these constraints not been redundant/tight, the operator would have had room to improve her profit by modifying the contracts without violating the remaining IR and IC constraints of honest buyers. 
%The proof is more involved in the presence of adversarial buyers, as we should also account for the change of behavior of the adversary under the modified contracts. 
For the contradiction argument to carry through, we show that under appropriate modifications, the effect of changes in the adversarial types' behavior on the revenue is non-decreasing. 
\end{proof}

We note that for the non-adversarial case, the same results of the above theorem holds; this follows from prior work in contract theory~\cite{maskin1984monopoly}  (using a straightforward mapping that we present in the  appendix). Formally:
\begin{proposition} \label{noadvcontract}
Assuming that the functions $b_i$ satisfy ID, the optimal contracts in the non-adversarial setting have $s_i^* = 0$ and satisfy all conditions of Theorem~\ref{advcontractproperties} (with $s_i^* = 0$). 
\end{proposition}

In particular, the relation between prices, fines, and benefit functions (points 3-4), provides an easy visual representation of the contracts as shown in Figure~\ref{fig:b-C-P} for the non-adversarial setting (that is, with $s_i =0$). We call this curve the \emph{price-contract} curve $\mathcal{P}(\epsilon)$, which is a curve on the $\epsilon$ (on x-axis), $p$ (on y-axis) plane, and connects the  non-adversarial contract points  $(\epsilon_i^*, p_i^*)$ for all $i$. From Proposition~\ref{noadvcontract}, we get $p_i^* - p_{i-1}^* = b_i(\epsilon_i^*) - b_i(\epsilon_{i-1}^*)$; thus, the segment of the curve $\mathcal{P}$ that is between $\epsilon^*_{i-1}$ and $\epsilon^*_i$ is parallel to $b_i(\cdot)$. Thus, $\mathcal{P}$ is continuous and piece-wise concave. 
%An analogous curve can be also defined for the adversarial setting, but in this work we use $\mathcal{P}$ to refer to the non-adversarial setting.

Theorem~\ref{advcontractproperties}'s characterization greatly simplifies the optimization problem to compute the optimal contracts by removing several of the constraints (points 1-2).
The result also shows that the optimization problem for computing optimal contracts in the presence of adversaries has only additional adversarial constraints and the same price-benefit relations (points 3-4) as that without adversaries. Despite these similarities, the presence of adversaries changes the seller's objective function, leading to a different set of contracts than the non-adversarial setting. Proposition~\ref{noadvcontract} further implies that the variables $s_i$ can be dropped in the optimization problem for the non-adversarial case, yet these variable remain a key design choice in the adversarial setting.

\medskip
\noindent\textbf{Price of Adversary}:
In order to quantify the effects of the adversary's presence on the seller's revenue, we introduce the following notion:
\begin{definition}[Price of Adversary]
Let $R^*$ and $R^*_A$ denote the seller's maximum revenue in non-adversarial and adversarial settings, respectively. %Let  denote the defender's maximum revenue in an adversarial setting. 
Then, the price of adversary ($PoAdv$) is defined as:
$$
PoAdv = (1-\rho)\frac{R^*}{R^*_A}
$$
\end{definition}
Clearly $PoAdv$ is $\geq 1$, with equality when the adversary does not choose any contract, so that $R^*_A = (1-\rho)R^*$. Our first finding is that $PoAdv$ is unbounded in the worst case.

\begin{lemma}\label{l:unboundedPoA}
$PoAdv$ is unbounded in the worst case.
\end{lemma}
\begin{proof}[Proof Sketch]
We prove this by construction with two types of legitimate users $H$ and $L$.
The benefit function are $b_L(\epsilon)=\log(1+\epsilon)$ for the lower type $L$ and $b_H(\epsilon)=2 \log(1+\epsilon)$ for the higher type $H$. 
%A legitimate user is of type $L$ (resp. $H$) with probability $q$ (resp. $1-q$).
The function for the adversary is given by $C(\epsilon)=(10/\rho + 2\frac{1- \gamma}{\rho \gamma} (\exp(\epsilon) -1)$.
We show that $R^*_A=0$, leading to an unbounded $PoAdv$.
\end{proof}

\section{Approximation Algorithm}
In this section, we present an approach that solves for the adversarial contracting problem approximately, given a solution for the non-adversarial case. We do so since solving the non-adversarial scenario is simpler: by Proposition~\ref{noadvcontract}, \rev{the non-adversarial case has both fewer variables ($s_i=0, \forall i$) and fewer constraints (no adversary contract choice constraint).} 
%all constraints except the limits on the contract variables can be removed for the non-adversarial case. 
 Our proposed  algorithm also reveals a subtle relation between the adversarial and non-adversarial settings.

Since by Lemma~\ref{l:unboundedPoA} we know that $PoAdv$ is unbounded in the worst case, we limit our analysis to a large class of adversary's benefit functions $C(\cdot)$ which imposes mild and natural restriction on these functions. We call these the \emph{well-behaved} $C$'s, and define them as follows. Recall that $\mathcal{P}$ denotes the non-adversarial price-contract curve  (see Figure~\ref{fig:b-C-P}).
\begin{itemize}
\item (High $ C$) $ C$ intersects $\mathcal{P}$ once at the origin and then lies above $\mathcal{P}$ for $\epsilon > 0$. 
\item (Low $ C$) $ C$ intersects $\mathcal{P}$ once at the origin and then lies below $\mathcal{P}$ for $\epsilon > 0$.
\item (Intermediate $ C$) $ C$ intersects $\mathcal{P}$ multiple times. Let $\epsilon_M\in(0,1)$ be the access level at the last intersection point. %; beyond $l_M$, $C$ lies above $\mathcal{P}$. Further, we assume that below $l_M$, \clar{we have $ C(l) - \mathcal{P}(l) \leq \delta$ for a small $\delta \geq 0$.} 
We denote $\Delta:=\max_{\epsilon<\epsilon_M}\{C(\epsilon) - \mathcal{P}(\epsilon)\}$. 
\end{itemize}

The above classes comprise several types of adversaries. High $C$s (low $C$s) represent powerful (weak) adversaries, who can (can not afford to) impose a high cost on the revenue; this class includes functions $C(\epsilon)\geq b_n(\epsilon)$ ($C(\epsilon)\leq b_1(\epsilon)$) as a subset. Intermediate $C$s on the other hand represent adversaries who can purchase (some of) the contracts offered through \emph{non-adversarial} contract design. Within this class, $\Delta$ is an upper bound on the adversary's payoff  from purchasing contracts with $\epsilon^*_i< \epsilon_M$.  As $C$ lies above $\mathcal{P}$ after $\epsilon_M$, we have $C(\epsilon^*_i) \geq p_i^*$ for all $\epsilon^*_i\geq \epsilon_M$, which means that the adversary can afford all contracts with $\epsilon^*_i\geq \epsilon_M$. Figure~\ref{fig:b-C-P} illustrates an intermediate $C$. Next, we present our approximation technique. We start with a definition.

\begin{definition} \label{def:slack}
We call the non-adversarial contract $(p, l, 0)$ a \emph{$\delta$-slack $\lambda$-priced} contract, ($\delta, \lambda \geq 0$), if there exists $s \geq 0$ such that the contract $(p-\gamma s,\epsilon,s)$ satisfies:
\begin{itemize}
    \item $C(\epsilon) - p - s \leq \delta$, i.e., adversary's gain is bounded by $\delta$.
    \item $p - \gamma s \geq \lambda > 0$, i.e., the contract's price is at least  $\lambda$.
    \item $p - \gamma s \geq (1 - \phi)p$, SR constraint is satisfied
\end{itemize}
Constructively, $s$ whenever it exists, should be chosen to have the least possible value.
\end{definition}

Using the above definition, our approximation technique is tailored towards the three categories of functions $C$ as shown in Algorithm~\ref{algorig}. This algorithm takes the set of non-adversarial contracts as input, and either successfully returns a new set of contracts by modifying this input, or prescribes solving the adversarial contract design problem from scratch. For the  High $C$ case, the algorithm finds $0$-slack contracts with a positive price (line 4, $0$-slack ensures the  adversary will not choose the new contract). If one is found, the contract generating the highest revenue among such contracts is offered to all users (line 10). For Low $C$, the adversary does not choose any contract, hence it is optimal to retain the non-adversarial contracts as is (line 13). For Intermediate $C$, the function $InterCApp$ presented in Algorithm~\ref{alg1} is invoked (line 15).
\begin{algorithm}[t]
\caption{Approx. Algorithm} \label{algorig} 
  \DontPrintSemicolon
  \SetAlgoNoEnd
  \SetKwComment{Comment}{$\triangleright$\ }{}
  \KwIn{Non-adv. contracts $(p_1^*, \epsilon_1^*,0), \ldots, (p_n^*, \epsilon_n^*,0)$}
\KwOut{An array of $contracts$ or solve adv. case}
$contracts \gets (p_1^*, \epsilon_1^*,0), \ldots, (p_n^*, \epsilon_n^*,0)$

\Switch{$C$}{
    \Case{High $C$}{
    $M = \{k~|~(p^*_i,\epsilon^*_i,0)$ is $0$-slack $\lambda$-priced for some $\lambda \geq 0\}$\;
    \If{$M$ is empty}{
    \Return{solve adv. case}
    }
    $j = \arg\!\max_{k \in M} p^*_k$\;
    $s^*_{j} \gets$ $s$ that makes $(p_{j}^*, \epsilon_{j}^*,0)$ $0$-slack $\lambda$-priced\;
    \For{$i\gets 1$ \KwTo $K$}{
    $contracts(i) = (p_{j}^* - \gamma s^*_j, \epsilon_j^*,s^*_j)$
    }
    \Return{$contracts$}
    }
    \Case{Low $C$}{
        \Return{$contracts$}
    }
    \Case{Intermediate $C$}{
        \Return{ Inter$C$App$((p_1^*, \epsilon_1^*,0), \ldots, (p_n^*, \epsilon_n^*,0))$
        }
    }
}
\Return{solve adv. case}
%  \Return{$contracts$}
\end{algorithm}

\begin{algorithm}[t]
\caption{$InterCApp$} \label{alg1}
  \DontPrintSemicolon
  \SetAlgoNoEnd
  \SetKwComment{Comment}{$\triangleright$\ }{}
  \KwIn{Non-adv. contracts $(p_1^*, \epsilon_1^*,0), \ldots, (p_n^*, \epsilon_n^*,0)$}
\KwOut{An array of $contracts$}
$K \gets $ highest $i$ such that $ \epsilon^*_{i} \leq \epsilon_M$\;
$E_{\geq K} = \{k~|~k \geq K$ and $(p_{k}^*, \epsilon_{k}^*,0)$ is $\Delta$-slack $p^*_K$-priced$\}$ \Comment*[r]{$E_{\geq K}$ not empty as $K \in E_{\geq K}$} 
 \For{$i\gets K+1$ \KwTo $n$}{
    $\mathsf{safe}(i) \gets \arg\!\max_{k \in E_{\geq K}} \{b_i(\epsilon^*_k) - p^*_k\}$\;
    $s^*_{\mathsf{safe}(i)} \gets$ $s$ that makes $(p_{\mathsf{safe}(i)}^*, \epsilon_{\mathsf{safe}(i)}^*,0)$ $\Delta$-slack $p^*_K$-priced\;
    $contracts(i) = (p_{\mathsf{safe}(i)}^* - \gamma s^*_{\mathsf{safe}(i)}, \epsilon_{\mathsf{safe}(i)}^*,s^*_{\mathsf{safe}(i)})$
    }
    \For{$i\gets 1$ \KwTo $K$}{
    $contracts(i) = (p_{i}^*, \epsilon_i^*,0)$
    }
    $\widehat{R}^*_{K} = \sum_{i=1}^{K} q_i p^*_i  + \sum_{i=K+1}^n q_i p_{\mathsf{safe}(i)}^*$\;
    $R^* = \sum_{i=1}^{n} q_i p^*_i $\;
    $\beta = \max\big(\max_{i \leq K} \{C(\epsilon^*_i) - p^*_i\}, \max_{i > K} \{C(\epsilon^*_{\mathsf{safe}(i)}) - p^*_{\mathsf{safe}(i)} - s^*_{\mathsf{safe}(i)}\} \big) $\;
    $\alpha = \max_i \{C(\epsilon^*_i) - p^*_i\}$\;
    
    \If{$(1- \rho)R^* - \rho\alpha  > (1- \rho)\widehat{R}^*_{K} - \rho\beta$}{
    \Return{$(p_1^*, \epsilon_1^*,0), \ldots, (p_n^*, \epsilon_n^*,0)$}
    }
  \Return{$contracts$}
\end{algorithm}

In Algorithm~\ref{alg1}, first a set of $\Delta$-slack $p_K^*$-priced contracts is found among contracts above and including that of type $K$ (line 2). The best contract with index $\mathsf{safe}(i)$ among these is found for each user $i > K$ (line 4). New contracts $(p_{\mathsf{safe}(i)}^* - \gamma s_{\mathsf{safe}(i)}, l_{\mathsf{safe}(i)}^*, s_{\mathsf{safe}(i)})$ are constructed for types $i > K$ (line 6), and all the non-adversarial contracts for types $K$ and below are retained as is (line 8). The revenue from honest buyers for the new contracts is found on line 9, and for the non-adversarial contracts on line 10. $\beta$ is the utility for the adversarial type in choosing the best new contract (line 11) and $\alpha$ is the same adversary utility in choosing from the non-adversarial contract set (line 12). Line 13-15 compares the revenue in the adversarial setting from the non-adversarial contracts and the new contract set, and returns the contract set that leads to better revenue for the seller. 

We next prove that the contracts output by Algorithm~\ref{algorig} are valid.
First, we present a lemma on the ordering of honest buyers' preferences over the contracts, which will later be used for the validity proof. 
\begin{lemma} \label{simple}
Given optimal non-adversarial contracts $(p_1^*, \epsilon_1^*,0), \ldots, (p_n^*, \epsilon_n^*, 0)$, a type $i$ user with $i > j$ prefers contract $(p^*_j, \epsilon^*_j, 0)$ over   $(p^*_k, \epsilon^*_k, 0)$ for $j > k$. 
%Also, a type $i$ user with $i < j$ prefers contract $(p^*_j, \epsilon^*_j, 0)$ over   $(p^*_k, \epsilon^*_k, 0)$ for $j < k$.
\end{lemma}
% \begin{proof}
% From Theorem~\ref{advcontractproperties}, we know that $b_j(\epsilon_j) - p_j = b_j(l_{j-1}) - p_{j-1}$, or equivalently, $b_j(\epsilon_j) - b_j(l_{j-1}) = p_j - p_{j-1}$. Using the ID property of the benefit functions, for $i> j$ we get $b_i(\epsilon_j) - b_i(l_{j-1}) \geq b_j(\epsilon_j) - b_j(l_{j-1}) = p_j - p_{j-1}$, hence $b_i(\epsilon_j) - p_j \geq b_i(l_{j-1})- p_{j-1}$. Thus, $i$ prefers contract $j$ to $j-1$. Arguing inductively, we have the required result. The case for $i < j$ follows exactly similarly but with the ID inequality flipped.
% \end{proof}

%We can now find the following bounds on the $PoAdv$: % within the class of well-behaved $C$'s. 

The validity of the Algorithm~\ref{algorig}'s output is as follows: %\com{by correctness, do we mean it is a valid  (although not necessarily optimal) solution to the optimization problem?}
\begin{lemma}
For Low or Intermediate $C$s, Algorithm~\ref{algorig}'s output contracts satisfy the IR and IC conditions for all honest buyers. If Algorithm~\ref{algorig} outputs a set of contracts for a High $C$ adversary, then at least one honest buyer buys the contract.
\end{lemma}
\begin{proof}
For High $C$, there is one contract offered to all users, so the IC constraints are trivially satisfied. Also, for user $j$ the contract offered satisfies IR, since from optimality of the non-adversarial contracts we get $b_j(\epsilon^*_j) - p^*_j \geq 0$.
For Low $C$, the proof is immediate from optimality of the non-adversarial contracts.%\com{is one sufficient? then my understanding of ``correctness'' is not right...}
For Intermediate $C$, if the non-adversarial contracts are returned by Algorithm~\ref{alg1}.  then the claim again holds trivially. Otherwise, if new contracts are returned, first, observe that the contract $(p^*_K, \epsilon^*_K, 0)$ is $\Delta$-slack $p^*_K$-priced (follows from Def.~\ref{def:slack} and definition of $K$, $\Delta$). Thus, $E_{\geq K}$ is not empty as $K \in E_{\geq K}$.
Also, note that for users $i > K$, the offered modified contracts (line 6) still has the effective price $p_{\mathsf{safe}(i)}^* - \gamma s_{\mathsf{safe}(i)}^* + \gamma s_{\mathsf{safe}(i)}^* = p_{\mathsf{safe}(i)}^*$ which is the  same as the non-adversarial contract.

We first start by analyzing users $i > K$. All users $j > K$ are offered modified contracts (line 6) among those indexed by $E_{\geq K}$ (loop on line 3). By definition of $\mathbf{safe}(i)$, $b_i(\epsilon^*_{\mathsf{safe}(i)}) - p_{\mathsf{safe}(i)}^* \geq b_i(\epsilon^*_{k}) - p_{k}^*$ for all $k \in E_{\geq K}$. Thus, $i$ prefers his contract over any other offered to any $j > K$. 
Next, by definition of $\mathsf{safe}(i)$, $b_i(\epsilon^*_{\mathsf{safe}(i)}) - p_{\mathsf{safe}(i)}^* \geq b_i(\epsilon^*_{K}) - p_{K}^*$, and then by Lemma~\ref{simple} and $i > K$, $b_i(\epsilon^*_{K}) - p_{K}^* \geq b_i(\epsilon^*_j) - p^*_j$ for all $j \leq K$. Thus, $b_i(\epsilon^*_{\mathsf{safe}(i)}) - p_{\mathsf{safe}(i)}^* \geq b_i(\epsilon^*_{j}) - p_{j}^*$ for all $j \leq K$. For IR, first by ID we have $b_i(\epsilon^*_{K}) \geq b_K(\epsilon^*_{K})$, hence $b_i(\epsilon^*_{K}) - p_{K}^* \geq b_K(\epsilon^*_{K}) - p_{K}^* \geq 0$, where the $\geq 0$ is due to optimality of the non-adversarial contracts. Finally, we just proved that $b_i(\epsilon^*_{\mathsf{safe}(i)}) - p_{\mathsf{safe}(i)}^* \geq b_i(\epsilon^*_{K}) - p_{K}^*$, thus, $b_i(\epsilon^*_{\mathsf{safe}(i)}) - p_{\mathsf{safe}(i)}^* \geq 0$.

The users $i \leq K$ are offered the non-adversarial contracts, thus, $b_i(\epsilon^*_{i}) - p_{i}^* \geq b_i(\epsilon^*_{j}) - p_{j}^*$ for all $j \neq i$. Since the modified contracts (line 6) still have an effective price same as the non-adversarial contract, any user $i \leq K$ still prefers his contract to the modified ones. The IR constraint is satisfied as the non-adversarial contracts were optimal.
\end{proof}

\begin{figure*}[t]
% \hspace*{-0.02in}%
\begin{minipage}[t]{0.24\textwidth}
%\centering
\includegraphics[width=0.99\linewidth,keepaspectratio=true]{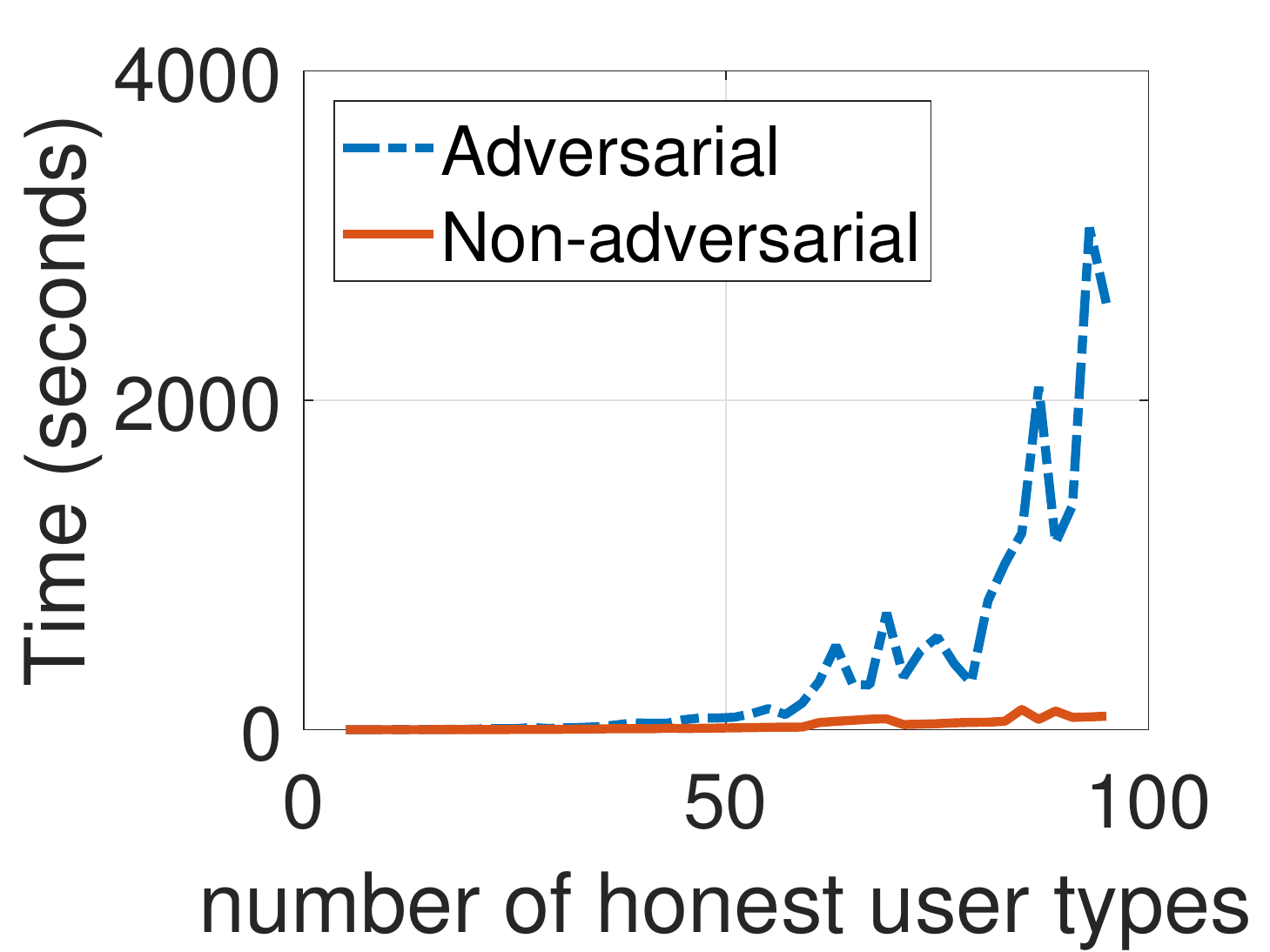}
\caption{Runtime comparison}
\label{fig:runtime}
\end{minipage}\hspace*{\fill}%
\begin{minipage}[t]{0.24\textwidth}
%\centering
\includegraphics[width=0.99\linewidth,keepaspectratio=true]{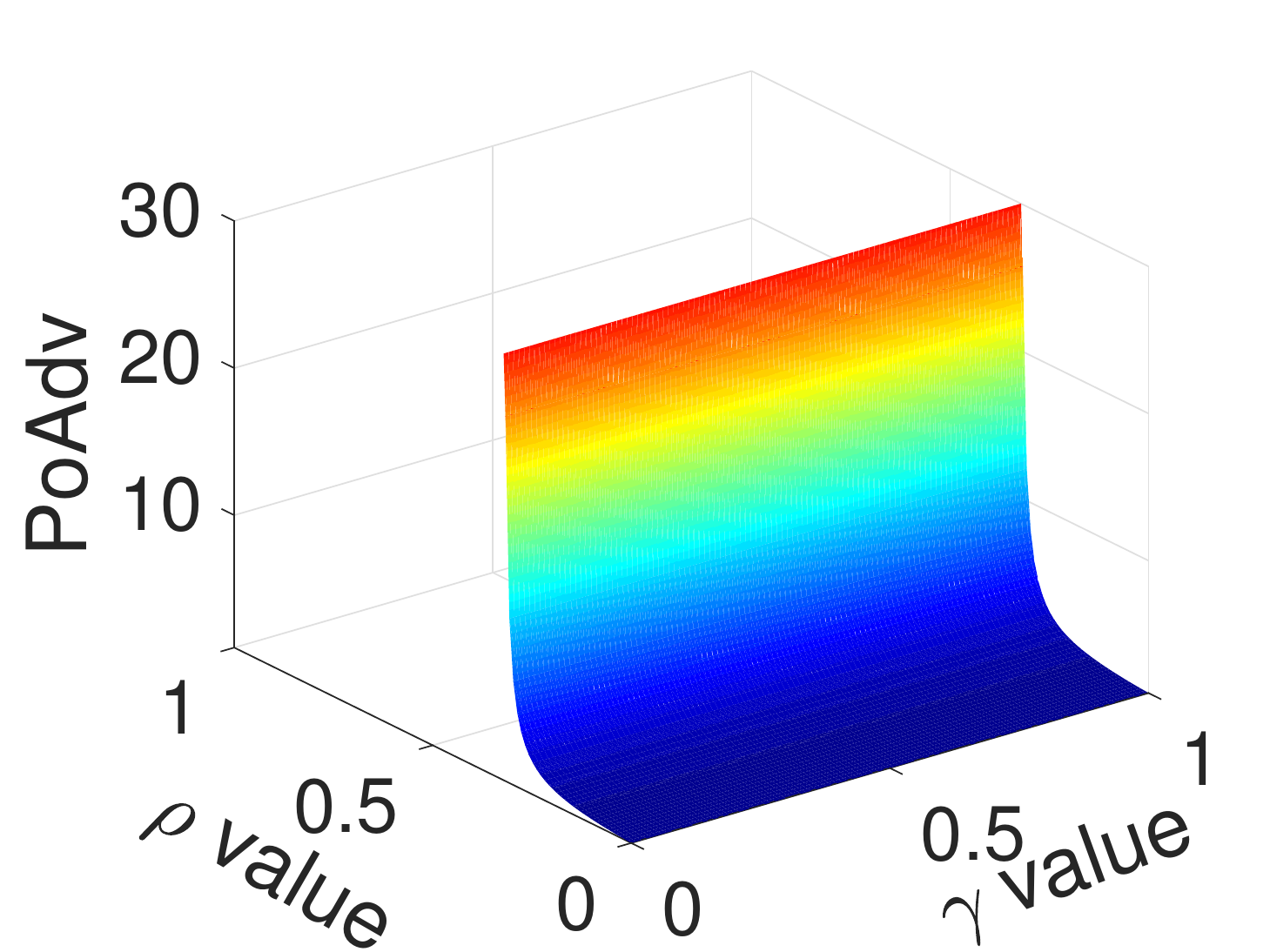}
\caption{$PoAdv$ for non-adversarial contracts}
\label{fig:noplan}
\end{minipage}\hspace*{\fill}%
\begin{minipage}[t]{0.24\textwidth}
%\centering
\includegraphics[width=0.99\linewidth,keepaspectratio=true]{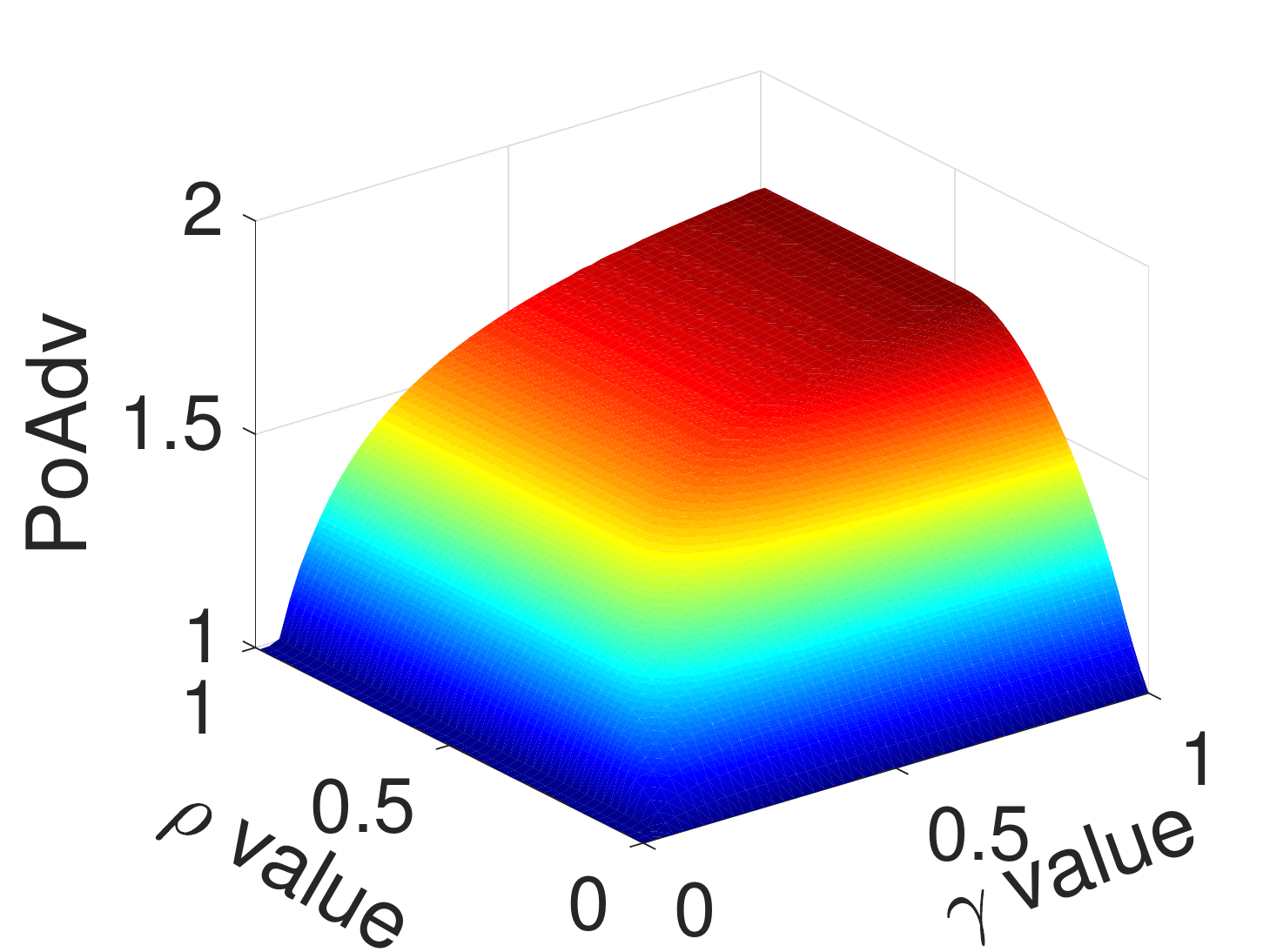}
\caption{$PoAdv$ for optimal adversarial contracts}
\label{fig:ideal}
\end{minipage}\hspace*{\fill}%
\begin{minipage}[t]{0.24\textwidth}
%\centering
\includegraphics[width=0.99\linewidth,keepaspectratio=true]{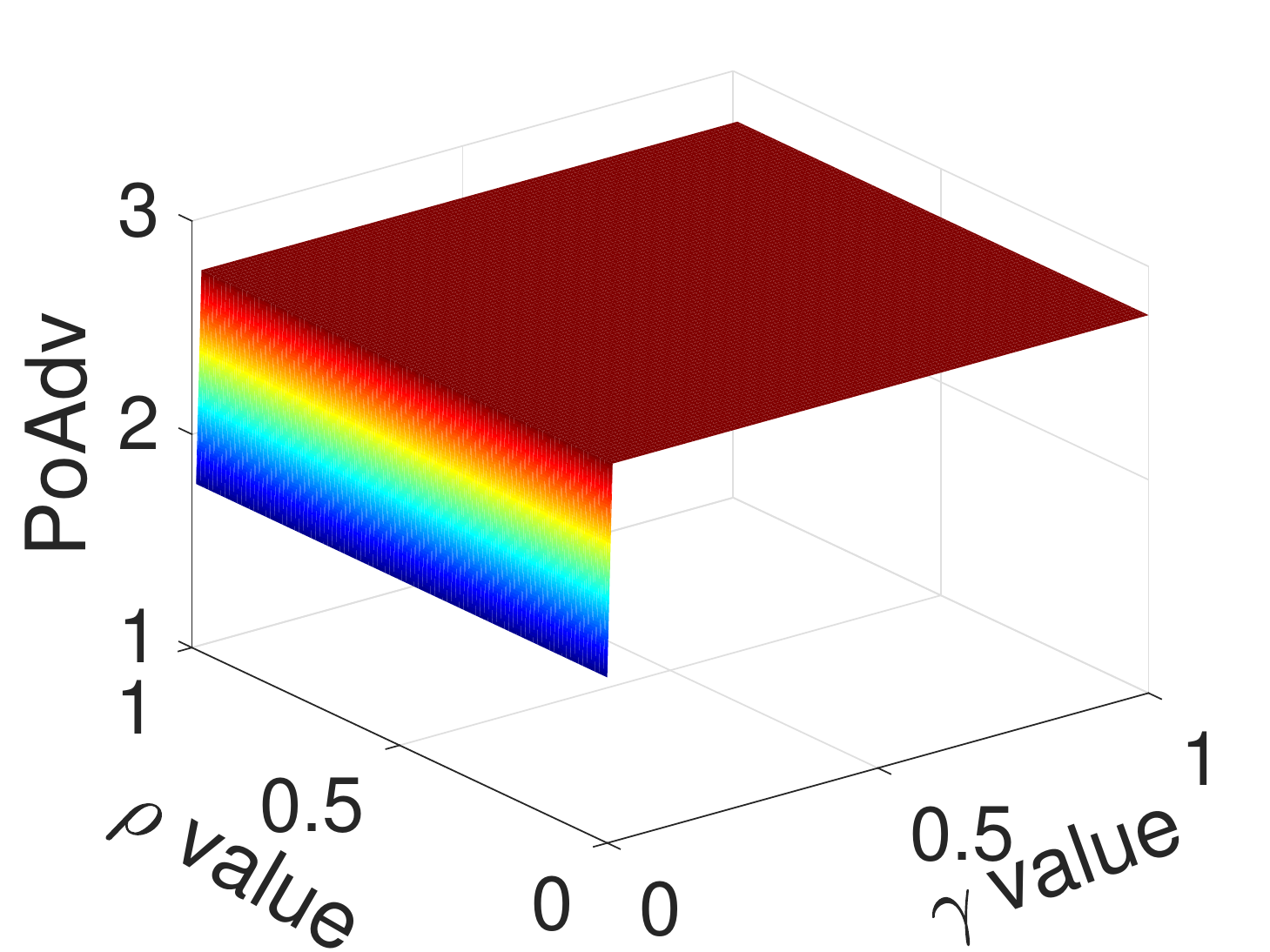}
\caption{$PoAdv$ for approx adversarial contracts}
\label{fig:approx}
\end{minipage}\hspace*{\fill}%
\end{figure*}

Next, the following result establishes the quality of the contracts returned by Algorithm~\ref{algorig} by bounding the $PoAdv$. Recall that we have already shown in Lemma~\ref{l:unboundedPoA} that $PoAdv$ is unbounded in the worst case.

\begin{theorem}\label{th:boundPoASDM}
Let the optimal non-adversarial contract $(p_1^*, \epsilon_i^*), \ldots, (p_n^*, \epsilon_n^*)$ revenue be $R^*$. For the class of well-behaved $C$'s, we have,
\begin{itemize}
\item (High $ C$) $PoAdv$ is unbounded in general. If Algorithm~\ref{algorig} outputs a  contract, then $PoAdv \leq \frac{R^*}{\lambda \min_i \{q_i\}}$. 
\item (Low $ C$) Algorithm~\ref{algorig} always outputs the same contracts as the non-adversarial case, and hence $PoAdv = 1$.
\item (Intermediate $ C$) Alg.~\ref{algorig} always outputs contracts. Then, 
$$PoAdv \leq \frac{R^*}{\max\Big(\widehat{R}^*_{K} - \Delta \frac{\rho}{1- \rho} , R^* - \alpha \frac{\rho}{1- \rho} \Big) } \; .$$
\end{itemize}
\end{theorem}
\begin{proof}[Proof Sketch] 
The analysis for High and Low $C$ is straightforward.
For intermediate $C$, depending on whether new contracts or the original contracts is offered, the revenue is max of the revenues in these two scenarios. Then the result follows by noting that $\beta \leq \Delta$.
\end{proof}

\section{Numerical Example}\label{sec:numerical}
While our theory results provide a broad characterization of the problem for a large space of utility functions,  
in this section we illustrate specific points related to the problem parameters, with a numerical example. We use $n=10$ types of honest buyers (except when varying $n$), with $b_i(\epsilon) = i(1 - \exp(-\frac{10\epsilon}{i}))$, 
%and for the adversary 
$C(\epsilon) = 6(\exp(\epsilon)-1)$, and $\phi=0.95$. 

\textbf{Runtime comparison}: Fig.~\ref{fig:runtime} illustrates runtimes when computing the optimal adversarial contracts and optimal non-adversarial contracts. The optimal adversarial contracts take much more time to compute than the non-adversarial contracts and the difference increases exponentially with increase in the size of problem $n$. This shows why approximation is useful; our approximation approach takes almost the same time as the non-adversarial problem, as the approximation steps after solving the non-adversarial problem have (comparatively) negligible runtime.

\textbf{Price of adversary with non-adversarial contracts}: Fig.~\ref{fig:noplan} shows the price of adversary for varying $\gamma$ and $\rho$ when the non-adversarial contracts are offered in an adversarial setting. We observe that the $PoAdv$ rises exponentially with $\rho$. Intuitively, the non-adversarial contracts suffer great losses if adversarial buyers dominate the market. % (i.e. $\rho$ increases).

\textbf{Price of adversary with optimal adversarial contracts}:
Fig.~\ref{fig:ideal} shows the price of adversary for varying $\gamma$ and $\rho$ when the optimal adversarial contracts are computed exactly. The $PoAdv$ rises with both increasing $\gamma$ and $\rho$. Intuitively, higher $\rho$ represents adversaries' market domination, and higher $\gamma$ is weaker honest users (i.e., more attack-prone). Thus, higher values for both of these parameters cause more loss, leading to higher $PoAdv$.

\textbf{Performance of approximation}: Lastly, Fig.~\ref{fig:approx} shows the price of adversary computed using our approximation approach for varying $\gamma$ and $\rho$. The $C$ that we chose corresponds to an Intermediate $C$.
The $PoAdv$ varies mostly with $\gamma$ and is almost constant throughout at 2.77, except for very small values of $\gamma$ when it is 1.43. For small values of $\gamma$, the approximation algorithm sends back the original contracts as is (line 14 in Algorithm~\ref{alg1}).

\section{Related Work}\label{sec:related}
Our work is within the emerging literature of data commercialization and its challenges \cite{thomas16}. %Both \cite{pantelis2013understanding} and \cite{thomas16} discuss the profit opportunities from packaging data based on the users' needs and willingness to pay; we formalize these notions through the framework of contract design with screening. 
%
%A number of recent papers have studied the design of optimal pricing mechanisms for data sellers. 
\cite{ghosh2015selling,gkatzelis2015pricing,li2014theory} study the problem of pricing personal data, where a data seller designs a pricing mechanism which incentivizes data subjects to reveal their private information. %, subject to factors such as these owners' privacy concerns and risk-aversion.
%\cite{balasubramanian2015pricing} compare the two pricing mechanisms of upfront payments and pay-per-use from the viewpoint of data sellers. 
\cite{niyato2016market} %study a market model
design a pricing scheme for selling data to users with differing willingness to pay. %, by optimally designing a pricing scheme. 
Our approach differs in that we propose a contract-theoretical framework to accommodate heterogeneous honest buyers as well as adversarial types. More specifically, in contrast to existing work, we posit that honest buyers do not attempt to compromise the privacy of the database, hence every sale of data is not a loss of privacy. Further, by far the practice in real world is for the data seller to obtain data by compensating people in form of a one-shot monetary payment or free service~\cite{datapay}, which is part of our model. This avoids unrealistic mechanisms in which data subjects are paid every time their data is sold to a buyer~\cite{li2014theory}. 

%Our work is also closely related to the literature  processing. 
%In an early survey on on privacy-preserving query by \cite{adam1989security}, a taxonomy of approaches into query restriction, data perturbation, and output perturbation was proposed, and drawbacks were pointed out for all of them.

\cite{adam1989security} classified privacy-preserving query approaches into query restriction, data perturbation, and output perturbation. Query auditing (a form of query restriction) aims to determine whether, given the query history, a new query will compromise the database privacy; however, this problem is NP-hard~\cite{kleinberg2003auditing}. In addition, output perturbation mechanisms (including differential privacy) must limit the number of queries in order to maintain any reasonable privacy guarantee~\cite{dinur2003revealing}. Our proposed approach, which is a combination of query restriction with output perturbation, restricts the type and number of queries in light of these impossibility results. 

Contract-theoretical frameworks have been receiving attention as a method for optimal pricing in other application areas, including the design of demand-response programs \cite{meir2017contract}, energy procurement methods \cite{tavafoghi2014optimal}, and incentive mechanisms in crowdsourcing markets \cite{ho2016adaptive}. In contrast, we consider the optimal pricing problem in the presence of both honest and adversarial users. 

Another line of work studies the effects of malicious or spiteful agents in game-theoretical settings, including network inoculation games \cite{moscibroda2006selfish}, sealed-bid auctions and colluding bidders \cite{brandt2007spiteful,micali}, and resource allocation games \cite{chorppath2011adversarial}. 
%Our proposed notion of price of adversary bears resemblance to \emph{price of malice} from this literature. 
These works assume that malicious agents aim to minimize the utility of all other users, and analyzes their effect on the Nash equilibria,  in a game-theoretic framework. In contrast, we consider the effects of an adversarial user on the  principal's revenue in a contract-theoretic framework.  

\section{Conclusion}\label{sec:conclusion}
We proposed a novel and practical \emph{adversarial contract design} framework in which a data seller designs a collection of contracts to optimize her revenue in the presence of honest and adversarial users. We quantified the effect of adversaries by proposing the notion of price of adversary, and characterized the effect of fines on optimal revenue. We also presented a fast approximate technique to compute contracts in an adversarial setting. 
%\com{Future directions include ... should we have this? what would be good?}

% \textbf{Online Appendix:} 
% Omitted proofs and additional numerical simulations are available at \url{https://www.dropbox.com/sh/8j1r5mlepmr1vow/AADM0bVJ8KSkSXtcoUG7JN2Pa?dl=0}

%\clearpage
\bibliographystyle{aaai}
\bibliography{privacy-contracts.bib}

%\end{document}

%\clearpage
\appendix

\section*{Appendix}

\subsection*{Theorem~\ref{advcontractproperties}}
\begin{proof}[Proof of Theorem~\ref{advcontractproperties}]
As a shorthand, we will write $p'_i = p_i + \gamma s_i$ throughout.

\textbf{Monotonicity:} First, we claim that for every optimal fixed cost contract we must have $\epsilon_i \geq \epsilon_j$ whenever $i > j$. Let $i > j$. The IC constraints include
$$
b_i(\epsilon_i) - p'_i  \geq b_i(\epsilon_j) - p'_j 
$$
$$
b_j(\epsilon_j) - p'_j  \geq b_j(\epsilon_i) - p'_i 
$$
Adding these, we get
$$
b_i(\epsilon_i) - b_i(\epsilon_j) \geq b_j(\epsilon_i) - b_j(\epsilon_j)
$$
There are two cases (1) $b_i(\epsilon_i) - b_i(\epsilon_j) > b_j(\epsilon_i) - b_j(\epsilon_j)$ or (2) $b_i(\epsilon_i) - b_i(\epsilon_j) = b_j(\epsilon_i) - b_j(\epsilon_j)$. For case (1), we can claim that $\epsilon_i \geq \epsilon_j$ using non-strict ID of the benefits functions. The proof is by contradiction. Assume $\epsilon_i < \epsilon_j$; then, by non-strict ID we must have $b_i(\epsilon_i) - b_i(\epsilon_j) \leq b_j(\epsilon_i) - b_j(\epsilon_j)$ which violates case (1). Hence under case (1) $\epsilon_i \geq \epsilon_j$. \adv{
As the reasoning here is not based on the seller's objective value or the adversarial type's constraints, we do not need to consider adversarial aspects here.}

Next, under case (2), let $K = b_i(\epsilon_i) - b_i(\epsilon_j) = b_j(\epsilon_i) - b_j(\epsilon_j)$. First, if $K$ is $\geq 0$ then $\epsilon_i \geq \epsilon_j$ when $b_i$ and $b_j$ are both strictly monotone increasing. \adv{As the reasoning here is not based on the objective value or the adversarial constraints, we do not need to consider adversarial aspects here.} The case when  $b_i$ and $b_j$ are both monotone non-decreasing has to be dealt in a special way (see after the $K<0$ case below). 

Thus, the only scenario left to analyze is $K < 0$. Then the two IC inequalities stated at the start can be re-written as $K \geq p'_i - p'_j$ and $-K \geq - (p'_i - p'_j)$ which implies $p'_i - p'_j = K$, or $p'_i < p'_j $. Also, $b_i(\epsilon_i) - p'_i = b_i(\epsilon_j) + K - p'_j - K = b_i(\epsilon_j) - p'_j$, so that contract $i$ and $j$ are both equally and \emph{most} preferred by $i$ (and similarly by $j$). Then offer
another set of contracts in which $i$ is offered $(p_j,\epsilon_j,s_j)$, and others are offered their earlier contract. In this new contract, all of the IC constraints are still satisfied as type $i$ preferred $(p_j,\epsilon_j,s_j)$ the most and equally preferred the now unavailable $(p_i,\epsilon_i,s_i)$. For any other type they prefer their allocation and price to $(p_j,\epsilon_j,s_j)$ as was the case for the earlier set of contracts. Also, since $b_i(\epsilon_i) -p'_i = b_i(\epsilon_j) - p'_j$ and earlier contract's IR provided $b_j(\epsilon_i) -p'_i \geq 0$, we have the new contract's IR is also satisfied $b_i(\epsilon_j) - p'_j \geq 0$. The $SR_i$ is also trivially satisfied since $SR_j$ was satisfied.
In this new set of contracts, as $p'_j > p'_i$ the revenue from $i$ increases and all other honest users provide same revenue as earlier, thus, the operator's revenue from the honest users strictly increases. \adv{Finally, we need to analyze the adversaries incentives in this new collection of contracts. For the adversary, the new set of contracts provides fewer options to choose from; thus, for any choice made by the adversary in the new contract regime, he obtains less or equal utility to that from the original contract set. As the operator's utility is zero-sum with the adversary's utility, the contribution from the adversarial part of the operator's revenue either increases or stays the same in the new set of contracts.} Thus, putting these together, we have found a new, feasible set of contracts, that strictly outperforms the original set of contracts, contradicting the optimality of the original set. Hence, we cannot have $K < 0$. 

\emph{Special case (non-decreasing $b_i$ and $b_j$)}: The case when  $b_i$ and $b_j$ are both monotone non-decreasing requires to treat the special case of $K=0$ separately. Thus, reasoning exactly like the $K<0$ case we get that $p'_i = p'_j$ and $b_i(\epsilon_i) = b_i(\epsilon_j)$ and $(p_i,\epsilon_i,s_i)$ and $(p_j,\epsilon_j,s_j)$ are both equally preferred by $i$. Now, if $\epsilon_i \geq \epsilon_j$ we are done, but if not we can offer $(p_j,\epsilon_j,s_j)$ to $i$. Following an argument similar to the case of $K<0$, the new set of contracts would satisfy all IR, SR, and IC constraints of the honest types. From the seller's viewpoint, the overall revenue from legitimate users remains the same as the original set of contracts. \adv{Further, following an argument similar to case $K<0$, the contribution from the adversarial part of the revenue either increases or stays the same with the new set of contracts.} 
Therefore, for this special case, we can claim that if $\epsilon_i<\epsilon_j$, the set of contracts is revenue equivalent (or even suboptimal to) a collection of contracts with $\epsilon_i=\epsilon_j$. We conclude that at the optimal contract, $\epsilon_i\geq \epsilon_j$ for this case as well.

{\bf Constraint-set reduction:} Next, we move on to the IC and IR constraints' properties. We start with the IR constraints. Starting from $IR_i$, we have, 
\begin{align*}
b_i(\epsilon_i) - p'_i &\geq b_i(\epsilon_{i-1}) - p'_{i-1}\\ 
&\geq b_{i-1}(\epsilon_{i-1})- p'_{i-1}~,
\end{align*}
where the first line follows from $IC_{i, i-1}$, and the second line by the assumption on ordering of the benefit functions, i.e., $b_i(l)\geq b_j(l), \forall i>j, \forall l$. Thus, if $IR_{i-1}$ is satisfied so is $IR_i$. Hence, given $IR_1$ is satisfied, all other IR constraints are redundant. \adv{
As the reasoning here is not based on objective value or adversary constraints, this assertion holds both with and without adversarial types.}

Next, we consider the (IC) constraints. 
By $IC_{i-1,i-2}$ we have $b_{i-1}(\epsilon_{i-1}) - p'_{i-1} \geq b_{i-1}(\epsilon_{i-2}) - p'_{i-2}$, which can be rearranged as $b_{i-1}(\epsilon_{i-1}) - b_{i-1}(\epsilon_{i-2})  \geq p'_{i-1} - p'_{i-2}$. 
By non-strict increasing difference and, as shown earlier, the monotonicity of access levels, $\epsilon_{i-1} \geq \epsilon_{i-2}$, we get,
\[b_{i}(\epsilon_{i-1}) - b_{i}(\epsilon_{i-2}) \geq b_{i-1}(\epsilon_{i-1}) - b_{i-1}(\epsilon_{i-2})  \geq p'_{i-1} - p'_{i-2}.\]
Thus, $b_{i}(\epsilon_{i-1}) - p'_{i-1} \geq b_{i}(\epsilon_{i-2}) - p'_{i-2}$. By $IC_{i,i-1}$, we have $b_i(\epsilon_i) - p'_i \geq b_i(\epsilon_{i-1}) - p'_{1-i}$, and hence we can infer that $b_i(\epsilon_i) - p'_i \geq b_{i}(\epsilon_{i-2}) - p'_{i-2}$. Thus, given the local downward IC constraints $IC_{i-1,i-2}$ and $IC_{i,i-1}$, the $IC_{i,i-2}$ constraint is redundant;  similarly, all $IC_{i, i -k}$ constraints are redundant for $k \geq 2$. 
Next, for the local upward IC constraints, starting from $IC_{i+1,i+2}$, we have $b_{i+1}(\epsilon_{i+1}) - p'_{i+1} \geq b_{i+1}(l_{i+2}) - p'_{i+2}$, which can be rearranged as $p'_{i+2} - p'_{i+1} \geq b_{i+1}(l_{i+2}) - b_{i+1}(\epsilon_{i+1})$. Again, by non-strict increasing difference and monotonicity $\epsilon_{i+1} \geq l_{i}$, we'll get $b_{i}(\epsilon_{i+1}) - p'_{i+1} \geq b_{i}(l_{i+2}) - p'_{i+2}$. Thus, we conclude that given the local upward IC constraints, all other upward constraints $IC_{i, i + k}$ for $k \geq 2$ are redundant. Hence, only the local $IC_{i, i +1}$ and $IC_{i, i -1}$ constraints are non-redundant. \adv{As the reasoning here is not based on objective value or adversary constraints, the arguments remain valid in the presence of adversaries.}

Next, we show that the local upward IC constraints $IC_{i, i + 1}$ is also redundant. For contradiction, suppose we solve the optimization problem without the $IC_{i,i+1}$ constraint, and get the set of contracts $\{p_j, \epsilon_j, s_j\}$ that maximize the operator's revenue. This solution should strictly violate $IC_{i,i+1}$ (since we are assuming $IC_{i,i+i}$ is not redundant). Therefore, type $i$ will strictly prefer the contract $\{p_{i+1}, \epsilon_{i+1}, s_{i+1}\}$, that is, $b_i(\epsilon_{i+1}) - b_i(\epsilon_i) > p'_{i+1} - p'_i$.  
We now modify the contracts by increasing $p_{j}, \forall j\geq i+1$ by a small amount $\epsilon>0$, i.e., we offer the contract $\{p_{i+1} + \epsilon, \epsilon_{i+1}, s_{i+1}\}$ to type $i+1$, as well as contracts $\{p_{j} + \epsilon, \epsilon_{j}, s_j\}$ for all $j > i+1$.  We chose $\epsilon$ small enough so that $IC_{i,i+1}$ remains strictly violated.

We know from the violation of  $IC_{i,i+1}$ that $b_i(\epsilon_{i+1}) - b_i(\epsilon_i) > p'_{i+1} - p'_i$, and also, by non-strict increasing differences, that $b_{i+1}(\epsilon_{i+1}) - b_{i+1}(\epsilon_i) > p'_{i+1} - p'_i$, or rearranging $b_{i+1}(\epsilon_{i+1}) - p'_{i+1}  > b_{i+1}(\epsilon_i) - p'_i$; thus, $IC_{i+1,i}$ is satisfied with $\{p_{i+1} + \epsilon, \epsilon_{i}, s_i\}$. 
For all other local upward IC constraints of types $i+1$ and higher (i.e, $IC_{i+1,i+2}$, $IC_{i+2,i+1}$, $IC_{i+2,i+3}$ and so on), the prices on both sides of the constraint change by an equal amount in the modified contract set. Therefore, these constraints continue to hold. For all other IC constraints there is no change in variable values and they continue to hold. The $IR_1$ constraint is also unaffected as the contract does not change for type $1$. All SR constraints still hold as only prices increased.
\adv{For the adversary, the contracts in the new collection are either the same (if he was purchasing one of the unaltered contracts) or become less attractive (if he was purchasing the altered contract). Thus, for any choice made by the adversary in the new contract regime, he obtains either less or the same utility as the original contract set. As the seller's and adversary's utilities are zero-sum, the contribution from the adversarial part of the revenue either increases or stays the same following the change in the contracts.}
Thus, this new set of contracts provides higher revenue to the operator, contradicting the optimality of the original set of contracts. We conclude that all local upward IC constraints should be redundant. 

% Next, for the $IR_i$ constraint, observe that by $IC_{i, i-1}$ and ID and the fact that $b_i(0) = 0$ for all $i$ we get: $b_i(\epsilon_i) - p'_i = b_i(\epsilon_{i-1}) - p'_{i-1} = b_i(\epsilon_{i-1}) - b_i(0) - p'_{i-1} \geq b_{i-1}(\epsilon_{i-1}) - b_{i-1}(0) - p'_{i-1} = b_{i-1}(\epsilon_{i-1}) - p'_{i-1}$. Thus, if $IR_{i-1}$ is satisfied so is $IR_i$. Hence, given $IR_1$ is satisfied all all IR constraints are redundant. {\color{red}
% As the reasoning here is not based on objective value or adversary constraints, we do not need to consider adversarial aspects here.
% } 

{\bf $IR_1$ is redundant, i.e., no information rent for the lowest type:} we prove this by contradiction. Suppose $IR_1$ is not binding; then, the operator can increase $p_1$ slightly without violating $IR_1$ (and trivially not violating $SR_1$). The only other constraint in which  in which $p_1$ appears is the LHS of the downward IC constraint $IC_{21}$. An increase in $p_1$ will lower the LHS, and hence this constraint will not be violated either. %All the leftover non-redundant legitimate user constraints never have $p_1$ on the LHS, and only have $-p_1$ on the RHS. Thus, increasing $p_1$ does not violate any legitimate user constraint. 
Therefore, the operator's portion of the revenue from legitimate users is strictly increasing with this increase in $p_1$. 
\adv{From the adversary's viewpoint, the new set of contracts (with an increased $p_1$ in the lowest type's contract) will either stay the same or becomes less attractive. Thus, for any choice made by the adversary in the new contract regime, he obtains less or equal utility to his utility in the original contract set. As the seller's revenue portion from the adversarial type's participation is the negative of the adversary's utility, the contribution from the adversarial part of the revenue will either increase or stay the same given the increase in $p_1$.} Thus, the modification of the price $p_1$ will lead to a feasible set of contracts that strictly increases the operator's revenue, contradicting the optimality of the original contract set. We thus conclude that $IR_1$ should be binding in the optimal contract set, so that $p_1^*+\gamma s_1^*=b_1(\epsilon_1^*)$. 

% Next, suppose $IR_1$ is not binding, then $p_1$ can be increased slightly. All the leftover non-redundant legitimate user constraints never have $p_1$ on the LHS, and only have $-p_1$ on the RHS. Thus, increasing $p_1$ does not violate any legitimate user constraint. {\color{red}
% For the adversary, for all offering in the new contract, the offering either stays the same or becomes less attractive for the adversary. Thus, for any choice made by the adversary in the new contract regime, he obtains less or same utility than the original contract regime. As the defender receives negation of the adversary's utility, the contribution from the adversarial part of the objective either increases or stays the same.} Thus, the revenue increases in the new offering contradicting that $p_i, \epsilon_i, s_i$ maximizes revenue.

{\bf $(IC_{i,i-1})$ is binding, i.e., information rent for higher types:} finally, we show that all the $IC_{i, i -1}, \forall i\geq 2$ constraints are binding. For contradiction, suppose $IC_{i, i - 1}$ is not binding. Then we can increase $p_{i}$ by $\epsilon$ without violating this constraint. In all remaining local downward IC constraint, $p_i$ only appears on the LHS of $IC_{i+1,i}$; the increase in $p_i$ will therefore not violate this constraint. %We can further increase $p_{i+1}... p_{n}$ by $\epsilon$. Since all these variable show up in pairs on opposite sides in the constraints $IC_{k, k + 1}$ for $k \geq i$ and $IC_{k, k - 1}$ for $k \geq i+1$, all these constraints remain satisfied after this change. Also, as $-p_i$ shows up on RHS of $IC_{i-1, i}$, increasing $p_i$ does not violate this constraint also. 
In addition, $IR_1$ will not be affected and also $SR_i$ constraint will not be violated as $p_i$ only appears on the LHS of $SR_i$, and the revenue of the operator from legitimate users will strictly increase following this change. \adv{For the adversary, for all contracts in the new set of contract, the contract either stays the same or becomes less attractive for the adversary (due to higher price). Thus, for any choice made by the adversary in the new contract regime, he obtains less or equal utility to that from the original contract set. As the seller receives the negative of the adversary's utility, the contribution from the adversarial part of the objective either increases or stays the same, and hence the overall revenue of the operator increases with the modified contract set.} This provides a contradiction to the optimality of the initial contracts. Therefore, the local downward IC constraints should be binding, leading to, 
$$p^*_{i+1}+\gamma s^*_{i+1} = b_{i+1}(l^*_{i+1}) - \sum_{j=1}^{i} \big( b_{j+1}(\epsilon^*_{j}) - b_{j}(\epsilon^*_j) \big )~.$$
In contract theory literature, the term $\sum_{j=1}^{i} \big( b_{j+1}(\epsilon^*_{j}) - b_{j}(\epsilon^*_j) \big )$ is known as the information rent of type $i$; it is the discount this type gets over the first-best prices that could be attained with full information about the users' types. 
\end{proof}

\subsection*{Proposition~\ref{noadvcontract}}
\begin{proof}[Proof of Proposition~\ref{noadvcontract}]
First, we will prove that for any optimal solution with non-zero $s_i$'s there is a revenue (objective) equivalent solution with all $s_i$ zero. The transformation is simple: given any solution $(p_i, \epsilon_i, s_i)$, the contract $(p_i + \gamma s, \epsilon_i, 0)$ is feasible and revenue optimal. The revenue stays the same, which trivially follows from the objective function. All constraints, except SR, have the term $p_i + \gamma s$, and hence they are satisfied. The SR constraints are trivially satisfied as $s_i$ is zero in the new contract.

Next, with contracts for which $s_i$ is 0, the optimization reduces to 
\begin{eqnarray*}
\max_{(p_i, \epsilon_i)_{i \in \Theta}} &  \sum_{i=1}^n q_i p_i  \\
\mbox{subject to} & IR_i \; \forall i \mbox{ and } IC_{i,j} \; \forall i,j \mbox{ and }  p_i,\epsilon_i \geq 0 \; \forall i 
\end{eqnarray*}
This is exactly same as the classic contract theory problem, and the conditions of Theorem~\ref{advcontractproperties} (with $s_i^* = 0$) follows from the seminal work by ~\cite{maskin1984monopoly}
\end{proof}

\subsection*{Lemma~\ref{l:unboundedPoA}}
\begin{proof}[Proof of Lemma~\ref{l:unboundedPoA}]
Consider a problem with two types of users $H$ and $L$.
Let the benefit function be $\log(1+\epsilon)$ for the lower type $L$ and $2 \log(1+\epsilon)$ for the higher type $H$. Given the user is not an adversary, the user is a lower type with probability $q$ and higher type with probability $1-q$.
The function $C$ for adversary is $K (\exp(\epsilon) -1)$, where $K$ will be chosen below. For now, let $K \geq 2 + 2(1-\gamma)/\gamma$.

For the adversarial revenue maximization case we will show the revenue is $0$. Let the contract with the adversary with $(p_L,\epsilon_L, s_L)$ and $(p_H,\epsilon_H,s_H)$. To show $0$ revenue we will show that $\epsilon_H = 0$ (and since $\epsilon_L \leq \epsilon_H$, $\epsilon_L = 0$). We do so by contradiction, whereby assume $\epsilon_H > 0$. First, it directly follows from Theorem~\ref{advcontractproperties} that $p_L + \gamma s_L = \log(1+\epsilon_L)$ and $p_H + \gamma s_H \leq 2\log(1+\epsilon_H)$. Next, as $\epsilon_H \geq \epsilon_L$, we have $p_L + \gamma s_L \leq \log(1+\epsilon_H)$.

The adversary never rejects the higher contract, since for any $\epsilon_H \in (0,1]$.
\begin{eqnarray*}
K(\exp(\epsilon_H) -1) &\geq& (2 + 2(1-\gamma)/\gamma)(\exp(\epsilon_H) -1) 
\\
& \geq & (2 + 2(1-\gamma)/\gamma)\log(1+\epsilon_H) \\
& \geq & 2\frac{1-\gamma}{\gamma}\log(1+\epsilon_H) + 2\log(1+\epsilon_H)\\
& \geq & 2\frac{1-\gamma}{\gamma}\log(1+\epsilon_H) + p_H + \gamma s_H
\end{eqnarray*}
Now, as $2\log(1+\epsilon_H) \geq  \gamma s_H$ (since $2\log(1+\epsilon_H) \geq p_H + \gamma s_H$ and $p_H \geq 0$), then the adversary does not reject the higher contract as 
$$2 \frac{1-\gamma}{\gamma}\log(1+\epsilon_H) + p_H + \gamma s_H \geq p_H + s_H$$
The above also provides the inequality
\begin{equation}
    (2 + 2(1-\gamma)/\gamma)\log(1+\epsilon_H) \geq p_H + s_H \label{temp123}
\end{equation}

Then, the adversary either chooses the lower or higher contract.
Then, the seller's utility is $(1-\rho)(q*(p_L + \gamma s_L) + (1-q)*(p_H + \gamma s_H)) + \rho[p_{Z} + s_Z - K(\exp(\epsilon_Z) -1)]$, where $Z$ is either $L$ or $H$. If $Z=H$, then the revenue is $(1-\rho)(q*(p_L + \gamma s_L) + (1-q)*(p_H + \gamma s_H)) + \rho[p_{H} +s_H - K(\exp(\epsilon_H) -1)]$. Then, observe that if $Z=L$, it means the adversary found $L$ more attractive, that is, $-[p_{L} + s_L - K(\exp(\epsilon_L) -1)] \geq -[p_{H} + s_H - K(\exp(\epsilon_H) -1)]$.  Thus, it can be said that $(1-\rho)(q*(p_L + \gamma s_L) + (1-q)*(p_H + \gamma s_H)) + \rho[p_{H} +s_H - K(\exp(\epsilon_H) -1)]$ is an upper bound on the revenue.

Next, $(1-\rho)(q*(p_L + \gamma s_L) + (1-q)*(p_H + \gamma s_H)) + \rho (p_{H} + s_H)$ must be less than $p_L + \gamma s_L + p_H + s_H$, which by previous inequality number~\ref{temp123} and $p_L + \gamma s_L \leq \log(1+\epsilon_H)$ is 
\begin{eqnarray*}
& \leq & \log(1+\epsilon_H) + (2 + 2(1-\gamma)/\gamma)\log(1+\epsilon_H)\\
& = & 3\log(1+\epsilon_H) + 2\frac{1-\gamma}{\gamma}\log(1+\epsilon_H)
\end{eqnarray*}
Now, choose 
$$K = 10/\rho + 2\frac{1-\gamma}{\rho\gamma} \; ,$$ which is clearly $\geq 2 + 2(1-\gamma)/\gamma$ also. Then, 
\begin{eqnarray*}
\rho K [\exp(\epsilon_H) -1] & = & (10 + 2(1-\gamma)/\gamma)* (\exp(\epsilon_H) -1)\\ & \geq & 10 \log(1+\epsilon_H) + 2\frac{1-\gamma}{\gamma}\log(1+\epsilon_H)
\end{eqnarray*}
Hence, the upper bound on revenue $(1-\rho)(q*(p_L + \gamma s_L) + (1-q)*(p_H + \gamma s_H)) + \rho[p_{H} +s_H] - \rho K(\exp(\epsilon_H) -1)]$ is less than $- 7 \log(1+\epsilon_H)$ which is strictly negative for positive $\epsilon_H$, and thus, the revenue is negative. This contradicts the optimality of $\epsilon_H$ as $0$ revenue is obtained with $p_H, \epsilon_H, s_H = 0$.
We conclude that $\epsilon_H=\epsilon_L=0$, so that $R^*_A=0$. 

On the other hand, without adversarial types,  the operator can attain positive revenue. This is because the seller's problem (using Theorem~\ref{advcontractproperties}) will be to maximize 
\begin{eqnarray*}
&& qb_L(\epsilon_L)  + (1-q)[b_H(\epsilon_H) - b_H(\epsilon_L) + b_L(\epsilon_L)] \\
&& = b_L(\epsilon_L) - (1-q) b_H(\epsilon_L) + (1-q)b_H(\epsilon_H)
\end{eqnarray*}
Hence, it is optimal to choose $\epsilon_H=1$, leading to a lower bound of $R^*\geq (1-q)2\log 2$ on the non-adversarial revenue. 
%On the other hand, without adversary the seller does not need any security deposits and has positive utility, for example, with $q=0.5$, the revenue is just $\log(1 + \epsilon_H)$, hence $\epsilon_H = 1$ yields the maximum revenue $\log 2$.
\end{proof}

\subsection*{Lemma~\ref{simple}}
\begin{proof}[Proof of Lemma~\ref{simple}]
%As a shorthand, we will write $p'_i = p_i + \gamma s_i$ throughout.
From Theorem~\ref{advcontractproperties}, we know that $b_j(\epsilon_j) - p_j = b_j(\epsilon_{j-1}) - p_{j-1}$, or equivalently, $b_j(\epsilon_j) - b_j(\epsilon_{j-1}) = p_j - p_{j-1}$. Using the ID property of the benefit functions, for $i> j$ we get $b_i(\epsilon_j) - b_i(\epsilon_{j-1}) \geq b_j(\epsilon_j) - b_j(\epsilon_{j-1}) = p_j - p_{j-1}$, hence $b_i(\epsilon_j) - p_j \geq b_i(\epsilon_{j-1})- p_{j-1}$. Thus, $i$ prefers contract $j$ to $j-1$. Arguing inductively, we have the required result. 
%The case for $i < j$ follows exactly similarly but with the ID inequality flipped.
\end{proof}

\subsection*{Theorem~\ref{th:boundPoASDM}}
\begin{proof}[Proof of Theorem~\ref{th:boundPoASDM}]
For High $C$, if a contract is offered, the adversary will not choose this contract due to $0$-slack, but at least user $i$ will choose it. Thus, the revenue is at least $(1-\rho) \lambda \min_i q_i $. For low $C$, clearly the adversary does not choose any contract, and $R^*_A = (1-\rho)R^*$.

For intermediate $C$, the revenue from the contracts output by Algorithm~\ref{algorig} is $\widehat{R}^*_K$. The adversary may choose any of the contracts. If the choice is $Z \leq K$ (all of which have fine $0$), then by definition of intermediate functions and $K$, $ C(\epsilon^*_Z) - p^*_{Z} \leq \Delta$ for all such $Z$. If  $Z > K$, then since the offered contracts $> K$ are all $\Delta$-slack, we again have $C(\epsilon^*_Z) - p^*_{Z} - s^*_Z\leq \Delta$. Thus, $\beta \leq \Delta$, and the revenue is lower bounded by $(1 - \rho)\widehat{R}^*_K - \rho \Delta$.
Finally, by not changing the original non-adversarial contracts the operator obtains a revenue $(1-\rho)R^* - \rho \alpha$. Thus, the revenue in the presence of adversaries is bounded by the maximum of either of these two lower bounds.
\end{proof}

\end{document}